%% file: main.tex
\documentclass[pdflatex,sn-basic]{sn-jnl}

\usepackage{amsmath,amssymb,array,graphicx,verbatim}

\usepackage{amsthm}
\usepackage{bbm}

\usepackage{relsize}
\usepackage{mathtools}
\usepackage{enumerate}
\usepackage[shortlabels]{enumitem}

\newtheorem{theorem}{Theorem}
\newtheorem{lemma}{Lemma}

\newtheorem{corollary}{Corollary}
\newtheorem{definition}{Definition}

\newtheorem{assumption}{Assumption}
\newtheorem{remark}{Remark}

\usepackage{enumerate}
\usepackage{epstopdf}
\usepackage{dsfont}
\usepackage{mathrsfs}

\DeclareMathOperator*{\argmax}{arg\,max\,}
\DeclareMathOperator{\ee}{\mathbb{E}}			
\DeclareMathOperator{\prob}{\mathbb{P}}			

\begin{document}

\title[Dynamic Games with Asymmetric Information and Hidden Actions]{An Approach to Stochastic Dynamic Games with Asymmetric Information and Hidden Actions} 

\author*[1]{\fnm{Yi} \sur{Ouyang}}\email{ouyangyii@gmail.com}

\author[2]{\fnm{Hamidreza} \sur{Tavafoghi}}\email{hamidreza.tavafoghi@gmail.com}

\author[3]{\fnm{Demosthenis} \sur{Teneketzis}}\email{teneket@umich.edu}

\affil[1]{\orgname{Preferred Networks America, Inc.}, \orgaddress{\city{Burlingame}, \state{CA}, \country{USA}}}

\affil[2]{\orgname{Google}, \orgaddress{\city{Mountain View}, \state{CA}, \country{USA}}}

\affil[3]{\orgdiv{Department of Electrical Engineering
and Computer Science}, \orgname{University of Michigan}, \orgaddress{\city{Ann Arbor}, \state{MI}, \country{USA}}}

\abstract{
We consider in discrete time, a general class of sequential stochastic dynamic games with asymmetric information with the following features. The underlying system has Markovian dynamics controlled by the agents' joint actions. Each agent's instantaneous utility depends on the current system state and the agents' joint actions. At each time instant each agent makes a private noisy observation of the current system state and the agents' actions in the previous time instant. In addition, at each time instant all agents have a common noisy observation of the current system state and their actions in the previous time instant. 
Each agent's actions are part of his private information. The objective is to determine Bayesian Nash Equilibrium (BNE) strategy profiles that are based on a compressed version of the agents' information and can be sequentially computed; such BNE strategy profiles may not always exist. We present an approach/methodology that achieves the above-stated objective, along with an instance of a game where BNE strategy profiles with the above-mentioned characteristics exist.
We show that the methodology also works for the case where the agents have no common observations.
}


\maketitle

\input{introduction}
\input{model}
\input{sufficient}
\input{example}
\input{specialcase}

\section{Conclusion}

We considered stochastic dynamic games where the underlying system is dynamic, the strategic agents' actions are hidden (not observable) and their information is asymmetric. We presented an approach for the computation of BNE strategy profiles that are based on a compressed version of the agents' information and can be determined sequentially in time moving backwards, if each step of this backward procedure has a solution. The approach highlights: (i) the importance of common information/common knowledge in identifying BNE strategy profiles that can be sequentially computed; (ii) the difference between common information that is sufficient for decision-making purposes in games and common information that is sufficient for decision-making purposes in teams. The difference is due to the fact that agents have an incentive to deviate from their predicted strategies in games whereas they don't have such an incentive in teams. As a consqence of this incentive, at each time instant each agent has his own view/belief of the game's status based on the common information, but all these different views/beliefs are common knowledge among all agents. As a result the CIB belief system is described by the sequence $\Pi^\psi_{1:T}$ specified by Definition \ref{def:CIB_belief_system}.

Our investigation focused on determining SIB-BNE strategy profiles for the games under consideration. We note that the SIB-BNE strategy profiles determined by our methodology are also Perfect Bayesian Equilibrium (PBE) strategy profiles when the agents have no common observations (i.e., for the model of Section \ref{sec:specialcase}), but this is not true when the agents have common observations (the general model of Section \ref{sec:model}). Determining PBE strategy profiles for the general model of Section \ref{sec:model} is an interesting problem worthy of investigation.

\input{appendix}
\bibliography{bib}
\end{document}

%% file: introduction.tex
\section{Introduction}

We study, in discrete time, a general class of sequential stochastic dynamic games with asymmetric information. We consider a setting where the underlying system has Markovian dynamics controlled by the agents’ joint actions. Each agent's instantaneous utility depends on the agents’ joint actions and the system state. At each time instant each agent makes a private noisy observation that depends on the current system state and the agents’ actions in the previous time instant. 
In addition, at each time instant all agents may have a common noisy observation of the system state and their actions in the previous time instant. 
The agents' actions are hidden, that is, each agent's actions are not directly observable by the other agents.
Therefore, at every time instant agents have asymmetric and imperfect information about the game's history.
Dynamic games with the above features arise in engineering (cybersecurity, transportation, energy markets), in economics (industrial organization), and in socio-technological applications.


As pointed out in \cite{tang2022dynamic}, the key challenges in the study of dynamic games with asymmetric information are: 
(i) The domain of agents' strategies increases with time, as the agents acquire information over time. Thus, the computational complexity of the agents' strategies increases with time. 
(ii) Due to signaling\footnote{
Signaling in games is more complex than signaling in teams because the agents have diverging incentives and their strategies are their own private information.
}
\citep{Ho:1980}, in many instances an agent's assessment of the game's status at time $t$, therefore his strategy at time $t$, depends on the strategies of agents who acted before him. Consequently, we cannot obtain the standard sequential decomposition (that sequentially determines the components of an equilibrium strategy profile) of the kind provided by the standard dynamic programming algorithm (where the agent's optimal strategy at any time $t$ does not depend on past strategies \cite[Chapter 6.5]{kumar1986stochastic}).

To address these challenges, we can look for equilibrium strategy profiles that are based on a compressed version of the agents' information and can be sequentially computed. However, such equilibrium strategy profiles may not exist.

In this paper we propose an approach, described in detail in Section \ref{sec:Methodology}, that addresses the above-stated challenges. According to this approach, we first compress the agents' private and common information at each time instant. Then, we define strategies based on the compressed information and show that Bayesian Nash Equilibria (BNE) based on these strategies can be determined sequentially in time moving backwards, if each step of this backwards procedure has a solution. Finally, we provide an example where a BNE strategy profile based on compressed information exists.

We show that the proposed approach works for the case where the agents have no common observations and their actions are hidden.


\subsection{Related Literature}

Dynamic games with asymmetric information have been extensively investigated in the literature in the context of repeated discounted games; see \cite{zamir1992repeated,forges1992repeated,aumann1995repeated,mailath2006repeated} and the references therein. The key feature of these games is the absence of a dynamic system. Moreover, the works on repeated games study primarily their asymptotic properties when the horizon is infinite and agents are sufficiently patient (i.e. the discount factor is close one). In repeated games, agents play a stage (static) game repeatedly over time. The main objective of this strand of literature is to explore situations where agents can form self-enforcing punishment/reward mechanisms so as to create additional equilibria that improve upon the payoffs they can get by simply playing an equilibrium of the stage game over time. Recent works (see \cite{horner2011recursive,escobar2013efficiency,sugaya2012}) adopt approaches similar to those used in repeated games to study infinite horizon dynamic games with asymmetric information when there is an underlying dynamic Markovian system. Under certain conditions on the system dynamics and information structure, the authors of \cite{horner2011recursive,escobar2013efficiency,sugaya2012} characterize a set of asymptotic equilibria attained when the agents are sufficiently patient.

The problem we study in this paper is different from the ones in \cite{zamir1992repeated,forges1992repeated,aumann1995repeated,mailath2006repeated,horner2011recursive,escobar2013efficiency,sugaya2012} in two aspects. First, we consider a class of dynamic games where the underlying system has general Markovian dynamics and a general information structure, and we do not restrict attention to asymptotic behaviors when the horizon is infinite and the agents are sufficiently patient. Second, we study situations where the decision problem that each agent faces, in the absence of strategic interactions with other agents, is a Partially Observed Markov Decision Process (POMDP), which is a complex problem to solve by itself. Therefore, reaching (and computing) a set of equilibrium strategies, which take into account the strategic interactions among the agents, is a very challenging task. As a result, it is not very plausible for the agents to seek reaching  equilibria that are generated by the formation of self-enforcing punishment/reward mechanisms similar to those used in infinitely repeated games. We believe that our results provide new insight into the behavior of strategic agents in complex and dynamic environments, and complement the existing results in the repeated games literature.

Stochastic dynamic zero-sum games with asymmetric information have been studied in \cite{renault2006value,cardaliaguet2015markov,gensbittel2015value,li2017solving,kartik2021upper,zheng2013decomposition,li2014lp}. The authors of \cite{renault2006value,cardaliaguet2015markov,zheng2013decomposition,li2014lp} study zero-sum games with Markovian dynamics and lack of information on one side (i.e. one informed and one uninformed agent).
The authors of \cite{gensbittel2015value,li2017solving,kartik2021upper} study zero-sum games with Markovian dynamics and lack of information on both sides. 
The works of \cite{renault2006value,cardaliaguet2015markov,gensbittel2015value,li2017solving,kartik2021upper,zheng2013decomposition,li2014lp} consider specific information structures.
Specifically: the actions of both agents are publicly observed; in \cite{renault2006value,cardaliaguet2015markov,zheng2013decomposition,li2014lp} the informed agent observes perfectly the state of the dynamic system, the other agent has no direct observation of the system's state; in \cite{gensbittel2015value,li2017solving} each agent observes perfectly part of the system's state and the states observed by the two agents are either independent or conditionally independent (given the observed actions). The authors of \cite{kartik2021upper} consider a general information structure where each agent has some private information and the agents share some information about the dynamic system's state and their actions. The authors of \cite{renault2006value,cardaliaguet2015markov,gensbittel2015value,li2017solving,kartik2021upper,zheng2013decomposition,li2014lp} derive their results by taking advantage of properties of zero-sum games such as the interchangeability of equilibrium strategies and the unique value of the game. These properties do not extend to non-zero sum games. We study a general class of stochastic dynamic games that include zero-sum stochastic dynamic games with asymmetric information as a special case. We consider general Markovian dynamics for the underlying system in contrast to \cite{renault2006value,cardaliaguet2015markov,gensbittel2015value,li2017solving,zheng2013decomposition,li2014lp}, where the system has the special structure described above. We consider a general information structure that allows us to capture scenarios with unobservable actions and imperfect observations that are not captured by \cite{renault2006value,cardaliaguet2015markov,gensbittel2015value,li2017solving,zheng2013decomposition,li2014lp}.

The problems investigated in \cite{tang2022dynamic, nayyar2014Common, gupta2014common, ouyang2015CDC, ouyang2016TAC, vasal2016signaling, sinha2016structured, gupta2016dynamic,nayyar2013common} are the most closely related to our problem. The authors of \cite{nayyar2014Common, gupta2014common,gupta2016dynamic,nayyar2013common} study a class of dynamic games where the agents’ common information based belief (defined in \cite{nayyar2014Common}) is independent of their strategies, that is, there is no signaling among them. This property allows them to apply ideas from the common information approach developed in \cite{nayyar2011optimal, nayyar2013decentralized}, and define an equivalent dynamic game with symmetric information among fictitious agents. Consequently, they characterize a class of equilibria for dynamic games called Common Information based Markov Perfect Equilibria. 

Our results are different from those in \cite{nayyar2014Common, gupta2014common,gupta2016dynamic,nayyar2013common} in two aspects. First, we consider a general class of dynamic games where the agents' CIB beliefs are strategy-dependent, thus, signaling is present. Second, the proposed approach in \cite{nayyar2014Common, gupta2014common,gupta2016dynamic,nayyar2013common} requires the agents to keep track of all of their private information over time. We propose an approach to effectively compress the agents’ private information, and consequently, reduce the number of variables which the agents need to form CIB beliefs.

The authors of \cite{tang2022dynamic, ouyang2015CDC, ouyang2016TAC, vasal2016signaling, sinha2016structured} study a class of dynamic games with asymmetric information where signaling occurs. When the horizon in finite, the authors of \cite{ouyang2015CDC, ouyang2016TAC} introduce the notion of Common Information Based Perfect Bayesian Equilibrium, and provide a sequential decomposition of the game over time. The authors of \cite{vasal2016signaling, sinha2016structured} extend the results of \cite{ouyang2015CDC, ouyang2016TAC} to finite horizon Linear-Quadratic-Gaussian (LQG) dynamic games and infinite horizon dynamic games, respectively. 

The work of \cite{tang2022dynamic} extends the model of \cite{ouyang2016TAC} to games among teams of agents. Each agent has his own private information which he shares with the members of his own team with delay $d$; teams also have common information. The authors of \cite{tang2022dynamic} consider two classes of strategies: sufficient private information based (SPIB) strategies, which only compress private information, and sufficient private and common information based (SPCIB) strategies, which compress both common and private information. They show that SPIB-strategy-based BNE exist and the set of payoff profiles of such equilibria is the same as the set of all BNE. They develop a backward inductive sequential procedure, whose solution, if it exists, provides a SPCIB BNE, and identify instances which guarantee the existence of SPCIB BNE. The class of dynamic games studied in 
\cite{tang2022dynamic, ouyang2015CDC, ouyang2016TAC, vasal2016signaling, sinha2016structured}
satisfy the following assumptions: (i) agents’ actions are observable (ii) each agent has a perfect observation of his own local states/type (iii) conditioned on the agents’ actions, the evolution of the local states are independent.
We relax assumptions (i)-(iii) of \cite{tang2022dynamic, ouyang2015CDC, ouyang2016TAC, vasal2016signaling, sinha2016structured}, and study a general class of dynamic games with asymmetric information, hidden actions, imperfect observations, and controlled and coupled dynamics. 

\subsection{Contribution}
We study/analyze, in discrete time, a general class of sequential stochastic dynamic games with asymmetric information, where the underlying system is dynamic, the information structure is non-classical, at each time instant the agents have private and common information and their actions are hidden (each agent's actions are not directly observable by the other agents). Our key contribution is a methodology for the discovery of Bayesian Nash Equilibrium (BNE) strategy profiles that are based on the agents' compressed private and common information and can be determined sequentially in time moving backwards, if each step of this backward procedure has a solution. We present an example where such a BNE strategy profile exists. 
We show that our methodology works also for the case where the agents have no common observations and their actions are hidden.

\subsection{Organization}
The rest of the paper is organized as follows: We present the game's model along with the equilibrium concept in Section \ref{sec:model}. We state our objective and present the methodology that achieves it in Section \ref{sec:Methodology}. In Section \ref{sec:compression} we first introduce compressed versions of the agents' private and common information that are sufficient for decision making purposes; then we define Sufficient Information Based (SIB) strategies that are based on the agents' compressed information. In Section \ref{sec:sequential_decomposition} we first introduce Sufficient Information Based Bayesian Nash Equilibrium (SIB-BNE); then we present a sequential decomposition of the game, that is, a backward inductive procedure that determines SIB-BNE if each step of this procedure has a solution. In Section \ref{sec:example} we present an example that highlights our solution methodology and where a SIB-BNE exists. In Section \ref{sec:specialcase} we show that our solution methodology works for stochastic dynamic games where the agents have no common observations and each agent's actions are part of his private information. 
The comparison of the definitions of compressed private information as it appears in this paper and in \cite{companion}, along with
some of the technical details related to the existence of SIB-BNE for the example of Section \ref{sec:example} are presented in the Appendices.

%% file: model.tex
\section{Model}
\label{sec:model}
We present our model for dynamic decision problems with strategic agents (dynamic games) below; this model is an analogue to the model of \cite{companion} for dynamic decision problems with non-strategic agents.

\subsection{System Dynamics} There are $N$ strategic agents who live in a dynamic Markovian world over horizon $\mathcal{T}\hspace*{-2pt}:=\hspace*{-2pt}\{1,2,...,T\}$, $T\hspace*{-2pt}<\hspace*{-2pt}\infty$. Let $X_t\hspace*{-2pt}\in\hspace*{-2pt}\mathcal{X}_t$ denote the state of the world at $t\hspace*{-2pt}\in\hspace*{-2pt}\mathcal{T}$. At time $t$, each agent, indexed by $i\hspace*{-2pt}\in\hspace*{-2pt} \mathcal{N}\hspace*{-2pt}:=\hspace*{-2pt}\{1,2,...,N\}$, chooses an action $a^i_t\hspace*{-2pt}\in\hspace*{-2pt}\mathcal{A}^i_t$, where  $\mathcal{A}^i_t$ denotes the set of available actions to him at $t$. Given the collective action profile $A_t\hspace*{-2pt}:=\hspace*{-2pt}(A_t^1,...,A_t^N)$, the state of the world evolves according to the following stochastic dynamic equation,\vspace*{-2pt}
\begin{align}
X_{t+1}=f_t(X_t,A_t,W_t^x), \label{eq:systemdynamic1} \vspace*{-2pt}
\end{align} 
where $W_{1:T-1}^x$ is a sequence of independent random variables. 
The initial state $X_1$ is a random variable that has a probability distribution $\mu_0\in\Delta(\mathcal{X}_1)$. 

At every time $t\in\mathcal{T}$, before taking an action, agent $i$ receives a noisy private observation $Y_t^i\in\mathcal{Y}_t^i$ of the current state of the world $X_t$ and the action profile $A_{t-1}$, given by\vspace*{-2pt}
\begin{align}
Y_t^i=O_t^i(X_t,A_{t-1},W_t^i), \label{eq:systemdynamic2}\vspace*{-2pt}
\end{align} 
where $W_{1:T}^i$, $i\in\mathcal{N}$, are sequences of independent random variables. Moreover, at every $t\in\mathcal{T}$, all agents receive a common observation $Z_t\in\mathcal{Z}_t$ of the current state of the world $X_t$ and the action profile $A_{t-1}$, given by\vspace*{-2pt}
\begin{align}
Z_t=O_t^c(X_t,A_{t-1},W_t^c), \label{eq:systemdynamic3}\vspace*{-3pt}
\end{align} 
where $W_{1:T}^c$, is a sequence of independent random variables. 
We assume that the random variables $X_1$, $W_{1:T-1}^x$, $W_{1:T}^c$, and $W_{1:T}^i$, $i\in\mathcal{N}$ are mutually independent. 

To avoid measure-theoretic technical difficulties and for clarity and convenience of exposition, we assume that all the random variables take values in finite sets.
\begin{assumption}\label{assump:finite}(finite game)
	The sets $\mathcal{N}$, $\mathcal{X}_t$, $\mathcal{Z}_t$, $\mathcal{Y}_t^i$, $\mathcal{A}_t^i$, $ i \in \mathcal N$, are finite.
\end{assumption}

\subsection{ Information Structure} 
Let $H_t$ denote the aggregate information of all agents at time $t$. Assuming that agents have perfect recall, we have $H_t=\{Z_{1:t},Y_{1:t}^{1:N},A_{1:t-1}^{1:N}\}$, \textit{i.e.} $H_t$ denotes the set of all agents' past and present observations and all agents' past actions. The set of all possible realizations of the agents' aggregate information is given by $\mathcal{H}_t:=\prod_{\tau\leq t}\mathcal{Z}_\tau\times\prod_{i\in\mathcal{N}}\prod_{\tau\leq t}\mathcal{Y}_\tau^i\times \prod_{i\in\mathcal{N}}\prod_{\tau< t}\mathcal{A}_\tau^i$. 

At time $t\hspace*{-2pt}\in\hspace*{-2pt}\mathcal{T}$, the aggregate information $H_t$ is not fully known to all agents. 
Let $C_t\hspace*{-2pt}:=\hspace*{-2pt}\{Z_{1:t}\}\hspace*{-2pt}\in\hspace*{-2pt}\mathcal{C}_t$ denote the agents' common information  about $H_t$ and $P_t^i\hspace*{-2pt}:=\hspace*{-2pt}\{Y_{1:t}^i,A_{1:t-1}^i\}\backslash C_t\hspace*{-2pt}\in\hspace*{-2pt}\mathcal{P}_t^i$ denote agent $i$'s private information about $H_t$, where $\mathcal{P}_t^i$ and $\mathcal{C}_t$ denote the set of all possible realizations of agent $i$'s private and common information at time $t$, respectively. 
We assume that observations $Y_\tau^i$, $\tau\in\{1,2...,t\}$, and actions $A_\tau^i$, $\tau\in\{1,2...,t-1\}$, are known to agent $i$ but are not necessarily fully known to all other agents, denoted by $-i$, at $t\in\mathcal{T}$. Therefore, we have $P_t^i\subseteq \{Y_{1:t}^i,A_{1:t-1}^i\}$ for all $i\in\mathcal{N}$, and $H_t=\left(\bigcup_{i\in\mathcal{N}}P_t^i\right)\cup C_t$ for all $t\in\mathcal{T}$. As such, $\left\{C_t,P_t^i,i\in\mathcal{N}\right\}$ form a partition of $\mathcal{H}_t$ at every time $t\in\mathcal{T}$. 
In Section \ref{sec:model:special}, we discuss several instances of information structures  that can be captured as special cases of our model.

\subsection{Strategies and Utilities:} Let $H_t^i:=\{C_t,P_t^i\}\in \mathcal{H}_t^i$ denote the information available to agent $i$ at $t$, where $\mathcal{H}_t^i$ denote the set of all possible realizations of agent $i$'s information at $t$. Agent $i$'s \textit{behavioral strategy} at $t$, denoted by $g_t^i$, is defined by
\begin{align}
    g^i_t:\mathcal{H}_t^i\rightarrow \Delta (\mathcal{A}_t^i)
    \label{eq:git}
\end{align}
where $\Delta (\mathcal{A}_t^i)$ is the set of Probability Mass Functions (PMFs) on $\mathcal{A}_t^i$.
We denote by 
\begin{align}
    g^i := (g^i_1,g^i_2,\ldots, g^i_T)
    \label{eq:gi}
\end{align}
a strategy of agent $i$; $g^i \in \mathcal G^i$, where $\mathcal G^i$ is the set of admissible strategies described by \eqref{eq:git}-\eqref{eq:gi}. 
We denote a strategy profile $g$ by 
\begin{align}
    g:= (g^1,g^2,\ldots,g^N)
    \label{eq:g}
\end{align}
$g \in \mathcal G$, where $\mathcal G$ is the set of admissible strategy profiles described by \eqref{eq:git}-\eqref{eq:g}. We denote by 
\begin{align}
    g^{-i}:= (g^1,\ldots,g^{i-1},g^{i+1},\ldots,g^N)
    \label{eq:gminusi}
\end{align}

Agent $i$'s instantaneous utility at $t$ depends on the system state $X_t$ and the collective action profile $A_t$, and is given by $u_t^i\hspace*{-1pt}(\hspace*{-1pt}X_t,\hspace*{-1pt}A_t\hspace*{-1pt})$. Agent $i$'s total utility over horizon $\mathcal{T}$, is given by,\vspace*{-2pt}
\begin{align}
U^i(X_{1:T},A_{1:T})=\sum_{t\in\mathcal{T}}u_t^i(X_t,A_t). \label{eq:totalutility}\vspace*{-2pt}
\end{align}

\subsection{Equilibrium Concept:} We consider Bayesian Nash Equilibrium (BNE) as the solution concept \citep{fudenberg1991game}. 
A strategy profile $g^* = (g^{*1}, g^{*2}, \ldots,g^{*N})$ is a BNE if for all $i \in \mathcal N$
\begin{align}
    \mathbb{E}^{g^*}\{U^i(X_{1:T},A_{1:T})\} \geq \mathbb{E}^{g^{*-i},\hat{g}^{i}}\{U^i(X_{1:T},A_{1:T})\},\quad\hspace{-4pt} \forall \hat{g}^i \in \mathcal G^i.
\end{align}

\subsection{Special Cases}\label{sec:model:special}
We discuss several instances of dynamic games with asymmetric information that are special cases of the general model described above.  

\vspace{3pt}

\textit{1) Nested information structure:} Consider a two-player game with one informed player and one uninformed player and general Markovian dynamics. At every time $t\hspace*{-2pt}\in\hspace*{-2pt} \mathcal{T}$, 
the informed player makes a private perfect observation of the state $X_t$, \textit{i.e.} $Y_t^1\hspace*{-2pt}=\hspace*{-2pt}X_t$. The uninformed player does not have any observation of the state $X_t$. Both the informed and uninformed players observe each others' actions, \textit{i.e.} $Z_t\hspace*{-2pt}=\hspace*{-2pt}\{A_{t-1}\}$. Therefore, we have $P_t^1=\{X_{1:t}\}$, $P_t^2=\emptyset$, and $C_t\hspace*{-2pt}=\hspace*{-2pt}\{A_{1:t-1}^1\hspace*{-1pt},\hspace*{-1pt}A_{1:t-1}^2\}$ for all $t\hspace*{-2pt}\in\hspace*{-2pt}\mathcal{T}$. The above nested information structure corresponds to dynamic games considered in \cite{renault2006value,cardaliaguet2015markov,renault2012value,li2014lp,li2017efficient,zheng2013decomposition}, where in \cite{renault2012value,li2017efficient} the state $X_t$ is static.

\vspace{3pt}

\textit{2) Delayed sharing information structure:} Consider a $N$-player game with observable actions where agents observe each others' observations with $d$-step delay. That is, $P_t^i=\{Y_{t-d+1:t}^i\}$ and $C_t=\{Y_{1:t-d},A_{1:t-1}\}$. We note that in our model we assume that the agents' common observation $Z_t$ at $t$ is only a function of $X_t$ and and $A_{t-1}$. Therefore, to describe the game with delayed sharing information structure within the context of our model we need to augment our state space to include the agents' last $d$ observations as part of the augmented state. Define $\tilde{X}_t:=\{X_t,M^1_t,M^2_t,...,M^d_t\}$ as the augmented system state where $M_t^i:=\{A_{t-i},Y_{t-i}\}\in\mathcal{A}_{t-i}\times\mathcal{Y}_{t-i}$, $i\in\mathcal{N}$; that is, $M_t^i$ serves as a temporal memory for the agents' observation $Y_{t-i}$ at $t-i$. Then, we have $\tilde{X}_{t+1}=\{X_{t+1},M_{t+1}^1,M_{t+1}^2,...,M_{t+1}^d\}=\{f_t(X_t,A_t,W_t^x),(Y_t),M_t^1,...,M_t^{d-1}\}$ and $Z_t=\{M_t^d, A_{t-1}\}=\{Y_{t-d}, A_{t-1}\}$.

The above environment captures a connection between the symmetric information structure and asymmetric information structure. The information asymmetry among the agents increases as $d$ increases. The above delayed sharing information structure corresponds to the dynamic game considered in \cite{tavafoghi2016stochastic}.

\vspace{5pt}

\textit{3) Perfectly controlled dynamics with hidden actions:} Consider a $N$-player game where the state $X_t\hspace*{-2pt}:=\hspace*{-2pt}(X_t^1\hspace*{-1pt},\hspace*{-1pt}X_t^2\hspace*{-1pt},\hspace*{-1pt}...,\hspace*{-1pt}X_t^N)$ has $N$ components. Agent $i$, $i\hspace*{-2pt}\in\hspace*{-2pt}\mathcal{N}$, perfectly controls $X_t^i$, \textit{i.e.} $X_{t+1}^i=A_t^i$. Agent $i$'s actions $A_t^i$, $t\hspace*{-2pt}\in\hspace*{-2pt}\mathcal{T}$, are not observable by all other agents $-i$. Every agent $i$, $i\hspace*{-2pt}\in\hspace*{-2pt}\mathcal{N}$, makes a noisy private observation $Y_i^t(X_t,W_t^i)$ of the system state at $t\hspace*{-2pt}\in\hspace*{-2pt}\mathcal{T}$. Therefore, we have $P_t^i\hspace*{-2pt}:=\hspace*{-2pt}\{A_{1:t},Y_{1:t}^i\}$, $C_t\hspace*{-2pt}=\hspace*{-2pt}\emptyset$.

\section{Objective and Methodology}
\label{sec:Methodology}

\subsection{Objective}
Our objective is twofold: (i) To determine BNE strategy profiles that are based on compressed versions of the agents' private and common information. (ii) To compute the above-mentioned strategy profiles by a sequential decomposition of the game, that is, by a backward inductive sequential procedure that identifies an equilibrium strategy profile when every step of the procedure has a solution.

\subsection{Methodology}
We present a methodology that achieves the above-state objective and proceeds as follows:
\begin{itemize}
    \item Step 1. We determine a mutually consistent compression of the agents' private information that is sufficient for decision-making purposes (such a mutually consistent compression may not be unique). Based on this compression we introduce the Sufficient Private Information Based (SPIB) belief system.
    \item Step 2. Based on the result of Step 1, we determine a compression of the agents' common information that is sufficient for decision-making purposes by defining the Common Information Based (CIB) belief system. The CIB belief system ensures that at each time instant each agent's CIB belief is consistent with his SPIB belief even when the agent deviates from his equilibrium strategy and plays an arbitrary strategy. Such a consistency implies that each agent forms his own CIB belief system, and each agent's CIB belief system is common knowledge among all agents.
    \item Step 3. Based on the compression of the agents' private and common information we introduce Sufficient Information Based (SIB) strategies for each agent (i.e., strategies that depend at each time on the agent's sufficient private information and the CIB belief system) and SIB BNE. We show that SIB strategies satisfy a key closedness of best response property. Based on this property we provide a sequential decomposition of the game, that is, a backward inductive sequential procedure that determines a SIB BNE if each step of the procedure has a solution.
    \item Step 4. We provide an example of a stochastic dynamic game with asymmetric information and hidden/unobservable actions where a SIB BNE exists.
\end{itemize}

%% file: sufficient.tex
\section{Compression of Private and Common Information}
\label{sec:compression}

In Section \ref{subsec:privatecompress}
we characterize/determine mutually consistent compressions of all agents' private information that are sufficient for decision-making purposes.
In Section \ref{subsec:SIBbelief} we introduce the common information based belief, a compressed version of the agents' common information, that is sufficient for decision making purposes.

\subsection{Sufficient private information (Step 1)}
\label{subsec:privatecompress}

We present/consider a compression of the agents' private information that is done in a mutually consistent manner so that the compressed information is sufficient for decision making purposes.

\begin{definition}[Sufficient private information]\label{def:sufficient}
We say that $S^i_t, i=1,\ldots,N$, is sufficient private information for the agents if
\begin{enumerate}[(i)]
\item $S^i_t$ is a function of $H^i_t$ such that
$S^i_t = \zeta^i_t(H^i_t)$  for some commonly known functions $\zeta^i_t, i=1,2,\ldots,N$.

\item $S^i_t$ can be sequentially updated as
$S^i_t = \phi^i_t(S^i_{t-1}, Y^i_t, Z_t, A^i_{t-1})$ using some commonly known functions $\phi^i_t,i=1,2,\ldots,N$.

\item For any realization $x_t, p^{-i}_t, p^i_t, c_t$, and the corresponding $s^{-i}_t =\zeta^{-i}_t(p^{-i}_t, c_t)$ and $s^i_t =\zeta^i_t(p^i_t, c_t)$, and any strategy profile $g$, where $g_t^i:\mathcal{S}_t^i\times C_t\rightarrow \Delta(\mathcal{A}_t^i),\forall i,\forall t$,  such that $\prob^g(p^i_t, c_t) > 0$,

\begin{align}
    \prob^g(x_t, s^{-i}_t \mid s^i_t, c_t) 
    = \prob^g(x_t, s^{-i}_t \mid p^i_t, c_t)
    \label{eq:sufficient-3}
\end{align}

\end{enumerate}
\end{definition}


\begin{remark}
A similar definition of sufficient private information for dynamic teams appears in \cite[Definition 2]{companion}. 
This definition is slightly different from Definition \ref{def:sufficient} above because the objectives in \cite{companion} and this paper are different.
In Appendix \ref{app:sufficient_information} we show that sufficient private information satisfying Definition \ref{def:sufficient} may violate condition (ii) of Definition 2 in \cite{companion}.
In \cite{companion} the compression of private (and common) information must entail no loss in performance, that is, we must be able to determine globally optimal team strategy profiles that are based on compressed private and common information. In this paper the goal is to determine BNE strategy profiles that are based on compressed information and be sequentially computed (if such BNE strategy profiles exist). We are not concerned about the equilibria we may lose when we compress information; therefore, we don't need condition (ii) of Definition 2 in \cite{companion}. 
\end{remark}

Definition \ref{def:sufficient} characterizes a set of compressions for agents' private information. In the following, we show the set of sufficient private information $S_t^i$, $i\in\mathcal{N}$, $t\in\mathcal{N}$, is rich enough to form belief systems on information sets of realizations with positive or zero probability.
Let $\tilde g^i$ denote the uniform strategy that assigns equal probability to every action of agent $i \in \mathcal N$.
Below we show that the policy-independence property of belief \cite[Theorem 1]{companion} for agent $i$ is still true when the private information $p_t^i$ is replaced with the sufficient private information $s_t^i$. That is, $\prob^{\tilde g^i, g^{-i}}(x_t, x^{-i}_t \mid s^i_t, c_t)$ constructed by $(\tilde g^i, g^{-i})$ captures agent $i$'s belief based on $h^i_t$ even when he plays an arbitrary strategy $\hat{g}^i$, not necessarily the same as $g^i$ or $\tilde g^i$, provided that agents $-i$ play $g^{-i}$.

\begin{lemma}
\label{prop:policy-independent}
For $h^i_t$ such that $\prob^{\hat g^i, g^{-i}}(h^i_t) > 0$, we have $\prob^{\tilde g^i, g^{-i}}(h^i_t) > 0$ and
\begin{align}
    \prob^{\hat{g}^i,g^{-i}}(x_t, s^{-i}_{t} \mid h^i_t) = \prob^{\tilde{g}^i, g^{-i}}(x_t, s^{-i}_{t} \mid h^i_t)    
    = 
    \prob^{\tilde{g}^i, g^{-i}}(x_t, s^{-i}_{t} \mid s^i_t, c_t) .
\end{align}
\end{lemma}
\begin{proof}
Note that $\prob^{\tilde g^i}(a^i_t) = 1 / \vert\mathcal A^i_t \vert$, so $\prob^{\tilde g^i, g^{-i}}(h^i_t) > 0$ given that $\prob^{g}(h^i_t) > 0$. Then
from part (i) of the definition of sufficient private information and part (i)  of Theorem 1 in \cite{companion} we have
\begin{align}
    \prob^{\hat{g}^i, g^{-i}}(x_t, s^{-i}_t \mid h^i_t) 
    = &\sum_{h^{-i}_{t}: \zeta^{-i}_t(h^{-i}_{t}) = s^{-i}_t }\prob^{\hat{g}^i, g^{-i}}(x_t, h^{-i}_{t} \mid h^i_t)
    \notag\\
    = & \sum_{h^{-i}_{t}: \zeta^{-i}_t(h^{-i}_{t}) = s^{-i}_t } \prob^{\tilde{g}^i, g^{-i}}(x_t,h_t^{-i} \mid h_t^i) 
    \notag\\
    = &
    \prob^{\tilde{g}^i, g^{-i}}(x_t, s^{-i}_{t} \mid h^i_t).
\end{align}
Furthermore, from condition (iii) of the definition of sufficient private information we have 
\begin{align}
    \prob^{\tilde{g}^i, g^{-i}}(x_t, s^{-i}_{t} \mid h^i_t)
    = \prob^{\tilde{g}^i, g^{-i}}(x_t, s^{-i}_{t} \mid s^i_t, c_t).
\end{align}
\end{proof}

\subsection{CIB Belief System (Step 2)}
\label{subsec:SIBbelief}

Given the compressed private information, we next compress the agents' common information in the form of a belief system.
We call such a compressed belief system the Common Information Based (CIB)  belief system.
Similar to \cite{tang2022dynamic, ouyang2016TAC}, the CIB belief system is sufficient for decision-making if it is common knowledge among all agents, and every agent $i$ can compute his belief about the system state and the other agents' sufficient private information using the CIB belief system and his compressed private information. More specifically, 
agent $i$ should be able to compute $\prob^{\hat g^i, g^{-i}}(x_t, s_t \mid h^i_t)$ using the CIB belief system and his sufficient private information $s^i_t$ whenever other agents follow the strategy profile $g^{-i}$ and agent $i$ plays an arbitrary strategy $\hat g^i$.


To determine a CIB belief system that satisfies the above sufficiency requirement we proceed as follows. We first define $N$ CIB belief systems $\Pi^\psi:= \{\Pi^{\psi, 1}, \Pi^{\psi, 2}, \ldots, \Pi^{\psi, N}\}$, one for each agent (Definition \ref{def:CIB_belief_system} below). Each belief system $\Pi^{\psi, i}$ consists of a sequence of PMFs on $\mathcal X_t \times \mathcal S_t$ that are sequentially updated according to an update rule $\psi=(\psi^1, \psi^2,\ldots,\psi^N)$ that is common knowledge among the agents; 
for each realization $c_t$ of the common information available at $t$, $\pi^{\psi,i}_t$ describes the belief on $\mathcal X_t \times \mathcal S_t$ based on $c_t$ from agent $i$'s point of view.
We want $\pi^{\psi,i}_t$, combined with $s^i_t$, to enable agent $i$ to form his own sufficient information-based private belief (given by $\prob^{\hat g^i, g^{*-i}}(x_t, s_t \mid s^i_t, c_t)$) about the current status of the game. Furthermore, we want the CIB belief system to capture the current status of the game when agents utilize strategies based on $(S_t, \Pi^\psi_t)$. For that matter, we define the notion/concept of Sufficient Information Based (SIB) strategy profile
$\sigma:= (\sigma^{i},i \in \mathcal N)$, $\sigma^i:= (\sigma^i_t, t \in \mathcal T), i \in \mathcal N$. Each component $\sigma^i_t$ of $\sigma$ is a function of $s^i_t$, agent $i$'s sufficient private information at $t$, and $\pi^\psi_t = (\pi^{\psi,i}_t, i \in \mathcal N)$ (see Definition \ref{def:SIB_strategy_n} below). Using the $N$ CIB belief systems and the SIB strategy profile $\sigma$ we define update equations for each $\pi^{\psi,i}_t$ so that each $\pi^{\psi,i}_t$ is consistent with $s^i_t$ and with agent $i$'s sufficient private information-based belief $\prob^{\hat g^i, g^{*-i}}(x_t, s_t \mid s^i_t, c_t)$, defined in Section \ref{subsec:privatecompress} (Definition \ref{def:sufficient}), and each $\pi^{\psi,i}_t$ is common knowledge among all agents (see Definition \ref{def:CIB_update_n} below). We proceed with the (formal) definitions.

\begin{definition}[Common information based (CIB) belief system]
\label{def:CIB_belief_system}

Given a sequence of update functions $\psi =\{\psi^{i}_t, i \in \mathcal N, t \in \mathcal T\}$
that are common knowledge among the $N$ agents, sequentially define
\begin{align}
    \Pi_{t}^{\psi,i} = \psi_{t}^i(\Pi^{\psi}_{t-1},Z_{t}), i \in \mathcal N, t \in \mathcal T
\end{align}
where 
\begin{align}
    &\Pi^{\psi}_t:= \left[\begin{array}{c}
     \Pi^{\psi, 1}_t  \\
     \vdots\\
     \Pi^{\psi, N}_t
    \end{array}\right], t \in \mathcal T
\\
    & \Pi^{\psi}_0:=  \left[\begin{array}{c}
     \mu_0  \\
     \vdots\\
     \mu_0
    \end{array}\right]
\end{align}
The sequence $\Pi^\psi_{1:T} = (\Pi^\psi_1, \Pi^\psi_2, \ldots, \Pi^\psi_T)$ defines a CIB belief system; $\Pi^{\psi, i}_t$ denotes the CIB belief over $\mathcal X_t \times \mathcal S_t$ based on $C_t$ from agent $i$'s point of view.
\label{def:SIB_belief_new_n}
\end{definition}

\begin{definition}[SIB strategy]
\label{def:SIB_strategy_n}
Given a CIB belief system $\Pi^\psi_{1:T}$, we define a Sufficient Information Based (SIB) strategy profile $\sigma:= (\sigma^1, \sigma^2,\ldots, \sigma^N)$, $\sigma^i:= (\sigma^i_1, \sigma^i_2, \ldots, \sigma^i_T)$ by the maps
\begin{align}
\sigma_t^i:\mathcal{S}^i_t\times [\Delta(\mathcal{X}_t\times \mathcal{S}_t)]^N \rightarrow \Delta(\mathcal{A}_t), 
t=1,2,\ldots, i=1,2,\ldots,N.
\end{align}
\end{definition}

Based on Definitions \ref{def:CIB_belief_system} and \ref{def:SIB_strategy_n} we present a set of conditions that an individual CIB belief system $(\Pi^{\psi,i}_t, t \in\mathcal T)$ must satisfy so as to ensure that each agent $i$ can form his own (private) belief about the current status of the game, given by $(X_t, S_t)$, using $\Pi^\psi_t$ and $S^i_t$ when all other agents $-i$ employ SIB strategies $\sigma^{-i}$. This set of conditions describe a sequential update rule of $\Pi^{\psi,i}_t$; the update rule depends on whether or not the (new) common observation at $t$ is feasible under the agents' strategies.

\begin{definition}[Consistent CIB belief system]
\label{def:CIB_update_n}
Consider a SIB strategy profile $\sigma$.
Let $F_{t}^i(x_{t+1},s_{t+1},z_{t+1})(\pi^\psi_t;\sigma^{-i}_t)$ denote the CIB belief about $(x_{t+1},s_{t+1},z_{t+1})$ constructed recursively by assuming that (i) $(x_t,s_t)$ is distributed according to $\pi_t^{\psi,i}$ (ii) agent $i$ employs the uniform strategy $\tilde g^i$ at $t$ (i.e., the strategy that chooses every action $a^i_t \in \mathcal A^i_t$ with equal probability), and (iii) agent $-i$ plays according $\sigma_t^{-i}$. That is,
\begin{align}
F_{0}^i(x_{1},s_{1},z_{1})
=& \sum_{y_{1}} \Bigg[\mathbb{P}\{z_{1},y_{1} \mid x_{1}\}\mu_0(x_{1})\left(\prod_{j}\mathbbm{1}\{s_{1}^j = \phi_{1}^j(z_{1},y_{1}^j)\}\right) \Bigg]
\end{align}
at $t=1$, and for $t \geq 1$.

\begin{align}
& F_{t}^i(x_{t+1},s_{t+1},z_{t+1})(\pi^\psi_t;\sigma^{-i}_t)
\notag\\
=& \sum_{y_{t+1},x_t,s_t,a_t} \Bigg[\mathbb{P}
\{z_{t+1},y_{t+1},x_{t+1} \mid  x_t,a_t\}\left(\prod_{j}\mathbbm{1}\{s_{t+1}^j = \phi_{t+1}^j(s_t^j,z_{t+1},y_{t+1}^j,a_t^j)\}\right)\nonumber\\
& 
\hspace{50pt} \left(\frac{1}{ \mid A_t^i \mid}\prod_{j\neq i} \sigma^j_t(a^j_t)(\pi^\psi_t,s_t^j) \right) \pi_t^{\psi,i}(x_t,s_t) \Bigg]
\end{align}

We define the update rule $\psi^\sigma = (\psi^{\sigma,i}_t, i \in \mathcal N, t \in \mathcal T)$ and the corresponding CIB belief system $\Pi^{\psi^\sigma}_{1:T}$ as follows. At any $t$ 

\begin{enumerate}[(i)]
    \item If $\sum_{\hat{x}_{t+1},\hat{s}_{t+1}} F_{t}^i(\hat{x}_{t+1},\hat{s}_{t+1},z_{t+1})(\pi^{\psi^\sigma}_t;\sigma^{-i}_t)>0$ (\textit{i.e.} the new common observation $z_{t+1}$ is feasible from the agent $i$'s point of view), then $\pi^{\psi^\sigma, i}_{t+1}$ can be updated recursively as
\begin{align}
\pi_{t+1}^{\psi^{\sigma},i}(x_{t+1},s_{t+1}) = \frac{F_{t}^i(x_{t+1},s_{t+1},z_{t+1})(\pi^{\psi^\sigma}_t;\sigma^{-i}_t)}{\sum_{\hat{x}_{t+1},\hat{s}_{t+1}} F_{t}^i(\hat{x}_{t+1},\hat{s}_{t+1},z_{t+1})(\pi_t^{\psi^\sigma};\sigma^{-i}_t)},
\label{eq:cib_bayesrule}
\end{align}
via Bayes rule.
\item If $\sum_{\hat{x}_{t+1},\hat{s}_{t+1}} F_{t}^i(\hat{x}_{t+1},\hat{s}_{t+1},z_{t+1})(\pi^{\psi^\sigma}_t;\sigma^{-i}_t)
=0$ (\textit{i.e.} the new common observation $z_{t+1}$ is infeasible from the agent $i$'s point of view), then the update rule is 
\begin{align}
\pi_{t+1}^{\psi^\sigma,i}(x_{t+1},s_{t+1}) = \frac{1}{\left\vert\mathcal{X}_{t+1}\times\mathcal{S}_{t+1}\right\vert}.
\label{eq:update-2_n}
\end{align}
\end{enumerate}
\end{definition}
Based on \eqref{eq:cib_bayesrule} and \eqref{eq:update-2_n} we can write

\begin{align}
    &\Pi_{t+1}^{\psi^\sigma,i} = \psi_{t+1}^{\sigma,i}(\Pi_t^{\psi^\sigma},Z_{t+1}).
    \label{eq:psi_it}
\\
    &\Pi_{t+1}^{\psi^\sigma} = \psi^{\sigma}_{t+1}(\Pi_t^{\psi^\sigma},Z_{t+1}).
    \label{eq:psi_t}
\end{align}

Furthermore, for all $i\in\mathcal N$, each agent can determine if 
$\sum_{\hat{x}_{t+1},\hat{s}_{t+1}} F_t^i(\hat{x}_{t+1},\hat{s}_{t+1},z_{t+1})(\pi_t^{\sigma^\psi};\sigma^{-i}_t)$ is positive or zero; thus each agent knows how agent $i$ computes $\pi^{\psi^\sigma, i}_{t+1}$ from $\sigma^i_t, z_{t+1}, \sigma^{-i}_t$ and $\psi^\sigma$. Therefore, $\pi^{\psi^\sigma, i}_{t}$ (hence $\pi^{\psi^\sigma}_{t}$) is common knowledge among all agents. We call $\Pi^{\psi^\sigma}_{1:T}$ the CIB belief system consistent with the SIB strategy profile $\sigma$.

\begin{remark}
Since the sufficient private information is a function of the agent's available information, a SIB strategy $\sigma^i_t$ corresponds to a strategy $g_t^{i, \sigma}$ given by
$g_t^{i, \sigma}(h^i_t) := \sigma^i_t(\zeta^i_t(h^i_t), \pi_t^{\psi^\sigma})$. 
Therefore, in the rest of the paper we use the following convention: $\prob^{\sigma}(\cdot) = \prob^{g^\sigma}(\cdot)$ and $\ee^{\sigma}[\cdot] = \ee^{g^\sigma}[\cdot]$.
\end{remark}

\begin{remark}
There are many alternative specifications of the update rule $\psi^\sigma_{t}, t \in \mathcal T$ defined by \eqref{eq:psi_it}-\eqref{eq:psi_t}, that result in consistent CIB belief systems, that is, CIB belief systems which ensure that (i) agent $i$ can form his private belief over $(X_t, S^{-i}_t)$ by incorporating his private sufficient information $S^i_t$ into his CIB belief $\Pi^{\psi^\sigma, i}_t$ given that agents $-i$ play according to $\sigma^{-i}$, (ii) agent $i$'s private belief formed according to $i$ is identical to the probability distribution over $(X_t, S^{-i}_t)$ conditional on his complete history $H^i_t$ even when he plays an arbitrary strategy $\hat g^i$ different from $\sigma^i$. An example of such an alternative update rule is described by \eqref{eq:cib_bayesrule} (Bayes' rule) when $\sum_{\hat{x}_{t+1},\hat{s}_{t+1}} F_t^i(\hat{x}_{t+1},\hat{s}_{t+1},z_{t+1})(\pi_t^{\psi^\sigma};\sigma^{-i}_t)>0$ and a arbitrary PMF $\pi^{\psi^\sigma, i}_{t+1}(\cdot, \cdot)$ on $X_{t+1} \times S_{t+1}$ when $\sum_{\hat{x}_{t+1},\hat{s}_{t+1}} F_t^i(\hat{x}_{t+1},\hat{s}_{t+1},z_{t+1})(\pi_t^{\psi^\sigma};\sigma^{-i}_t)=0$.
\end{remark}

Definition \ref{def:CIB_update_n} ensures that agent $i$ can form his beliefs over $(X_t,S_t^{-i})$ by incorporating his sufficient private information $S_t^i$ into his CIB belief $\Pi_t^{\psi^\sigma,i}$ given that agents $-i$ play according to $\sigma^{-i}$. 
Moreover, this belief is sufficient to compute the probability distribution over $(X_t,S_t^{-i})$ conditional on his complete history $H_t^i$ even when he plays an arbitrary strategy $\hat{g}^i$ different from $\sigma^i$. 
We formalize the above discussion in Lemma \ref{lemma:CIBbelief-privatebelief} below, by using the notation 
$\prob^{\hat g^i, \sigma^{-i}, \psi^\sigma}(\cdot)$ to indicate the belief resulting when agent $i$ plays $\hat g^i$ and agents $-i$ play $g^{-i,\sigma}(h^{-i}_t)=\sigma^{-i}_t(\zeta^{-i}_t(h^{-i}_t), \pi_t^{\psi^\sigma})$ using the update rule $\psi^\sigma$
.

\begin{lemma}\label{lemma:CIBbelief-privatebelief}
Consider a SIB strategy profile $\sigma$, along with an associated consistent CIB belief system $\Pi_t^{\psi^\sigma}$. Suppose $(x_t, h^i_t, h^{-i}_t$ is a realization with positive probability under $(\hat g^i, \sigma^{-i})$, where $\hat g^i$ denotes an arbitrary strategy for agent $i$. 
Let $s^i_t = \zeta^i_t(h^i_t)$ and $s^{-i}_t = \zeta^{-i}_t(h^{-i}_t)$ be the associated sufficient private information.
Then agent $i$'s belief at time $t$ can be computed using $\pi_t^{\psi^\sigma}$ as
\begin{align}
    \prob^{\hat g^i, \sigma^{-i}, \psi^\sigma}
    (x_t, s^{-i}_t \mid h^i_t)
    = \frac{\pi^{\psi^\sigma,i}_t(x_t, s_t)}
        {\sum_{s^{-i}_t, x_t}\pi^{\psi^\sigma,i}_t(x_t, s^i_t, s^{-i}_t)}
\label{eq:lemma_cibbelief}
\end{align}

\end{lemma}

\begin{proof}
From Lemma \ref{prop:policy-independent}
we have
\begin{align}
    \prob^{\hat g^i, \sigma^{-i}, \psi^\sigma}(x_t, s^{-i}_t \mid h^i_t)
    = \prob^{\tilde g^i, \sigma^{-i}, \psi^\sigma}(x_t, s^{-i}_t \mid h^i_t) .
    = \prob^{\tilde g^i, \sigma^{-i}, \psi^\sigma}(x_t, s^{-i}_t \mid c_t, s^i_t).
    \label{eq:lemma_cibbelief_pf1}
\end{align}
By Bayes' rule we obtain
\begin{align}
 \prob^{\tilde{g}^i, \sigma^{-i}, \psi^\sigma}(x_t, s^{-i}_t \mid c_t, s^i_t) =
 \frac{\prob^{\tilde{g}^i, \sigma^{-i}, \psi^\sigma}(x_t, s_t \mid c_t)}{\prob^{\tilde{g}^i, \sigma^{-i}, \psi^\sigma}(s^{i}_t \mid c_t)} =  \frac{\pi^{\psi^\sigma,i}_t(x_t, s_t)}
        {\sum_{s^{-i}_t, x_t}\pi^{\psi^\sigma,i}_t(x_t, s^i_t, s^{-i}_t)}.
        \label{eq:lemma_cibbelief_pf2}
\end{align}
Combination of \eqref{eq:lemma_cibbelief_pf1} and \eqref{eq:lemma_cibbelief_pf2} establishes the assertion of Lemma \ref{lemma:CIBbelief-privatebelief}.
\end{proof}

\begin{remark}
\label{remark:condindepend}
Suppose $X_t = (X^1_t, X^2_t,.\ldots, X^N_t)$ and we have the conditional independence property, namely, that for any strategy profile $g$ $\prob^g(x_t, s_t \mid c_t) = \prod_i \prob^{g^i}(x^i_t, s^i_t \mid c_t)$. Then one can show for any $i$ that
\begin{align*}
 \pi^{\psi^\sigma,i}_t(x_t, s_t) = \prod_j \pi^{\psi^\sigma,i}(x^j_t, s^j_t) =  
\prob^{\tilde g^{i}_t}(x^i_t, s^{i}_t \mid c_t)
\prod_{j \neq i} \prob^{\sigma^{j}}(x^j_t, s^{j}_t \mid c_t)
\end{align*}
Therefore, for settings with the conditional independence property as in \cite{tang2022dynamic, ouyang2016TAC}, one can use the simplified beliefs $\prob^{\tilde g^{i}_t}(x^i_t, s^{i}_t \mid c_t)$ and $\prob^{\sigma^{j}}(x^j_t, s^{j}_t \mid c_t)$ as the compressed common information to compute the CIB belief $\pi^{\psi^\sigma, i}_t(x_t,s_t)$.
The conditional independence among the system components in the models of \cite{tang2022dynamic, ouyang2016TAC} could be lost when the agents' actions are not observable. 
\end{remark}

\section{
Sequential decomposition (Step 3)
}
\label{sec:sequential_decomposition}

In this section we present a sequential decomposition of the game, that is, a backward inductive sequential procedure that determines a Sufficient Information Based Bayesian Nash Equilibrium (SIB-BNE), defined below, if each step of this procedure has a solution.
We proceed as follows. We first establish a key closedness of best response property (Section \ref{subsec:closedness}); we use this property to provide a sequential decomposition of the game (Section \ref{subsec:dynamicprogram})

\begin{definition}[SIB-BNE]
\label{def:SIBBNE_original}

Consider a SIB strategy profile $\sigma^* = (\sigma^{*1}, \sigma^{*2}, \ldots, \sigma^{*n})$ and its corresponding consistent update rule $\psi^{\sigma^*}$.
The SIB strategy profile $\sigma^*$ is a SIB-BNE if it is a BNE of the dynamic game. That is, for all $i \in \mathcal N$,
\begin{align}
    \mathbb{E}^{\hat{g}^{i},\sigma^{*-i}, \psi^{\sigma^*}}\{U^i(X_{1:T},A_{1:T})\}
    \leq
    \mathbb{E}^{\sigma^*, \psi^{\sigma^*}}\{U^i(X_{1:T},A_{1:T})\},
    \notag\\
    \text{ for all strategies (not necessarily SIB strategies) }\hat{g}^i.
\end{align}

\end{definition}

\subsection{Closedness of best response}
\label{subsec:closedness}

The key result of this subsection is presented in the following theorem.
\begin{theorem}
\label{thm:closedness}
Consider a fixed and known SIB strategy profile $\sigma$ and the corresponding update rule $\psi^\sigma$. Suppose agents $-i$ use $\sigma^{-i}$ with $\psi^\sigma$. Then, there exists a SIB strategy $\hat\sigma^{i}$ that uses $\psi^\sigma$ and is a best response to $\sigma^{-i}$ with $\psi^\sigma$.
\end{theorem}

The proof is based on Lemmas \ref{thm:POMDP_new_n}, \ref{lemma:markov_new_n}, and \ref{lemma:utility_new_n} that we state and prove below.


\begin{lemma}
\label{thm:POMDP_new_n}

Consider a SIB strategy profile $\sigma$ and the corresponding update rule $\psi^\sigma$ along with the consistent CIB belief system $\Pi_{1:T}^{\psi^\sigma}$.

If agents $-i$ play according to the SIB strategies $\sigma^{-i}$ and use the update rule $\psi^\sigma$, the best response problem for agent $i$ is a POMDP with state and observation processes
\begin{align}
    & \tilde X_t = (S_t, \Pi_t^{\psi^\sigma}, X_t), t \in \mathcal T\\
    & \tilde Y_t = (Y^i_t, Z_t), t \in \mathcal T
\end{align}
respectively, and instantaneous utility
\begin{align}
\tilde u^i_t(\tilde X_t, A^i_t)
= 
\sum_{a^{-i}_t} 
    \big(\prod_{j \neq i}
    \sigma^{j}_{t}(a^j_t \mid S^j_t, \Pi_t^{\psi^\sigma})
    \big)
    u^i_t(X_t, a^{-i}_t, A^i_t), t \in \mathcal T
\end{align}
\end{lemma}

The assertion of Lemma \ref{thm:POMDP_new_n} is a direct consequence of Lemmas \ref{lemma:markov_new_n} and \ref{lemma:utility_new_n}.

\begin{lemma}
\label{lemma:markov_new_n}
Consider a SIB strategy profile $\sigma$ and the corresponding update rule $\psi^\sigma$. Suppose agents $-i$ play according to the SIB strategies $\sigma^{-i}$ using $\psi^\sigma$ and agent $i$ follows an arbitrary strategy $\hat{g}^i$ (not necessarily a SIB strategy). Then
\begin{align}
    \prob^{\hat{g}^i,\sigma^{-i},\psi^\sigma}(\tilde x_{t+1}, \tilde y_{t+1} \mid 
    \tilde x_{1:t}, \tilde y_{1:t}, a^i_{1:t})
    = \prob^{\hat{g}^i\sigma^{-i},\psi^\sigma}(\tilde x_{t+1}, \tilde y_{t+1} \mid 
    \tilde x_{t}, a^i_{t})
\end{align}
\end{lemma}
\begin{proof}
The probability for the next state and observation $\tilde x_{t+1}, \tilde y_{t+1}$ can be computed by
\begin{align}
    & \prob^{\hat{g}^i,\sigma^{-i},\psi^\sigma}(\tilde x_{t+1}, \tilde y_{t+1} \mid 
    \tilde x_{1:t}, \tilde y_{1:t}, a^i_{1:t})
    \notag\\
    = & \prob^{\hat{g}^i,\sigma^{-i},\psi^\sigma}(x_{t+1}, \pi^{\psi^\sigma}_{t+1}, s_{t+1}, y^i_{t+1}, z_{t+1}\mid
    x_{1:t}, \pi^{\psi^\sigma}_{1:t}, s_{1:t}, y^i_{1:t}, z_{1:t}, a^i_{1:t})
    \notag\\
    = & \sum_{y^{-i}_{t+1}, a^{-i}_t} \prob^{\hat{g}^i,\sigma^{-i},\psi^\sigma}(
    x_{t+1}, \pi^{\psi^\sigma}_{t+1}, s_{t+1}, y_{t+1}, z_{t+1}, a^{-i}_t\mid 
    x_{1:t}, \pi^{\psi^\sigma}_{1:t}, s_{1:t}, y^i_{1:t}, z_{1:t}, a^i_{1:t})
    \notag\\
    = & \sum_{y^{-i}_{t+1}, a^{-i}_t} 
    \big(\prod_{j}\mathds{1}(s^j_{t+1} = \phi^j_{t+1} (s^j_t, y^j_{t+1}, z_{t+1}, a^j_t))
    \big)
    \mathbb{P}\{z_{t+1},y_{t+1}, x_{t+1}\mid x_{t},a_{t}\}
    \notag\\
    & \hspace{2cm}
    \mathds{1}(\pi^{\psi^\sigma}_{t+1} = \psi^\sigma_{t+1}(\pi^{\psi^\sigma}_t,z_{t+1}))
    \big(\prod_{j \neq i}
    \sigma^{j}_{t}(a^j_t \mid s^j_t, \pi^{\psi^\sigma}_t)
    \big)
    \label{eq:markov_update_1_n}
\end{align}
where the last equality follows from the system dynamics, part (ii) of Definition \ref{def:sufficient}, Definition \ref{def:CIB_update_n}, and the form of SIB strategies of agents ${-i}$.
Since the right hand side of \eqref{eq:markov_update_1_n} depends only on $(\tilde x_{t}, a^i_{t})$
we conclude that 
\begin{align}
    \prob^{\hat{g}^i,\sigma^{-i},\psi^\sigma}(\tilde x_{t+1}, \tilde y_{t+1} \mid 
    \tilde x_{1:t}, \tilde y_{1:t}, a^i_{1:t})
    = \prob^{\hat{g}^i,\sigma^{-i},\psi^\sigma}(\tilde x_{t+1}, \tilde y_{t+1} \mid 
    \tilde x_{t}, a^i_{t})
\end{align}
\end{proof}

Lemma \ref{lemma:markov_new_n} shows that $\{\tilde{X}_t, \tilde{Y}_t, t \in \mathcal T\}$ is a Markov process conditional on $\{A^i_t, t \in \mathcal T\}$

\begin{lemma}
\label{lemma:utility_new_n}
Consider a SIB strategy profile $\sigma$ and the corresponding update rule $\psi^\sigma$. Suppose agents $-i$ follow the SIB strategies $\sigma^{-i}$ using $\psi^\sigma$ and agent $i$ follows an arbitrary strategy $\hat{g}^i$ (not necessarily a SIB strategy).
Then there are utility functions $\tilde u^i_t$ such that
$\ee^{\hat{g}^i,\sigma^{-i},\psi^\sigma}[\tilde u^i_t(\tilde X_t, A^i_t)] = \ee^{\hat{g}^i,\sigma^{-i},\psi^\sigma}[u^i_t(X_t, A_t)]$ for all $t\in\mathcal T$.
\end{lemma}
\begin{proof}
Recall that $\tilde X_t = (S_t, \Pi^{\psi^\sigma}_t, X_t)$. Then
\begin{align}
    & \ee^{\hat{g}^i,\sigma^{-i},\psi^\sigma}[u^i_t(X_t, A_t)] 
    \notag\\
    = & \ee^{\hat{g}^i,\sigma^{-i},\psi^\sigma}[u^i_t(X_t, A^{-i}_t, A^i_t)]
    \notag\\
    = & \ee^{\hat{g}^i,\sigma^{-i},\psi^\sigma}\big[
    \ee^{\hat{g}^i,\sigma^{-i},\psi^\sigma}[u^i_t(X_t, A^{-i}_t, A^i_t)\mid \tilde X_t, A^i_t]\big]
    \notag\\
    = & \ee^{\hat{g}^i,\sigma^{-i},\psi^\sigma}\big[
    \sum_{a^{-i}_t} \prob^{\hat{g}^i,\sigma^{-i},\psi^\sigma}(a^{-i}_t \mid S_t, \Pi^{\psi^\sigma}_t, X_t, A^i_t) u^i_t(X_t, a^{-i}_t, A^i_t)]
    \big]
    \notag\\
    = & \ee^{\hat{g}^i,\sigma^{-i},\psi^\sigma}\big[
    \sum_{a^{-i}_t} 
    \big(\prod_{j \neq i}
    \sigma^{j}_{t}(a^j_t \mid S^j_t, \Pi^{\psi^\sigma}_t)
    \big)
    u^i_t(X_t, a^{-i}_t, A^i_t)]
    \big]
\end{align}
Therefore, we establish the claim of the lemma by defining 
\begin{align}
\tilde u^i_t(\tilde X_t, A^i_t) = \sum_{a^{-i}_t} 
    \big(\prod_{j \neq i}
    \sigma^{j}_{t}(a^j_t \mid S^j_t, \Pi^{\psi^\sigma}_t)
    \big)
    u^i_t(X_t, a^{-i}_t, A^i_t)]
\end{align}
\end{proof}

\begin{proof}[Proof of Theorem \ref{thm:closedness}]
From Lemma \ref{thm:POMDP_new_n} we conclude that the best response of agent $i$ to $\sigma^{-i}$ is a POMDP with state $\tilde X_t$. From the theory of POMDP \cite[Chapter 6]{kumar1986stochastic} we know that: (i) the belief on the state $\tilde X_t = (S_t, \Pi^{\psi^\sigma}_t, X_t)$ conditioned on available information $h^i_t$ is an information state for the agent; (ii)
for each $t \in \mathcal T$ there exists an optimal strategy for agent $i$ that is a function of the information state at $t$.
We now prove that $(S^i_t, \Pi^{\psi^\sigma}_t)$ is an information state for agent $i$ at $t, t \in \mathcal T$.
\\
We note that $S^i_{t+1}=\phi^i_t(S^i_t, Y^i_{t+1}, Z_{t+1}, A^i_t)$ from part (ii) of Definition \ref{def:sufficient}, and $\Pi^{\psi^\sigma}_{t+1} = \psi^\sigma_{t+1}(\Pi^{\psi^\sigma}_t, Z_{t+1})$ from \eqref{eq:psi_t}.

Thus, we only need to show that 
for any strategy $\hat{g}^i$ and any realization $h^i_t$ such that $\prob^{\hat{g}^i, \sigma^{-i},\psi^\sigma}(h^i_t) > 0$ the following equality is true:
\begin{align}
    \prob^{\hat{g}^i, \sigma^{-i},\psi^\sigma} (s_t, \pi^{\psi^\sigma}_t, x_t \mid h^i_t)
    = 
    \prob^{\hat{g}^i, \sigma^{-i},\psi^\sigma} (s_t, \pi^{\psi^\sigma}_t, x_t \mid s^i_t, \pi^{\psi^\sigma}_t)
    \label{eq:thmclosedness_pf1}
\end{align}
For that matter, we note that $s^i_t, \pi^{\psi^\sigma}_t$ are perfectly known to agent $i$. Furthermore, from the definition of sufficient private information and Lemma \ref{lemma:CIBbelief-privatebelief} we have
\begin{align}
    \prob^{\hat{g}^i, \sigma^{-i},\psi^\sigma} (s^{-i}_t, x_t \mid h^i_t)
    = \frac{\pi^{{\psi^\sigma}, i}_t(s_t, x_t)}
        {
        \sum_{s^{-i}_t, x_t}\pi^{{\psi^\sigma}, i}_t(s^i_t, s^{-i}_t, x_t)},
        \label{eq:thmclosedness_pf2}
\end{align}
which is a function of $(s^i_t,\pi^{\psi^\sigma}_t)$. Therefore, 
\begin{align}
    \prob^{\hat{g}^i, \sigma^{-i},\psi^\sigma} (s_t, \pi^{\psi^\sigma}_t, x_t \mid h^i_t)
    = \mathds{1}(s^i_t = \zeta^i_t(h^i_t))\mathds{1}(\pi^{\psi^\sigma}_t = \gamma^{\psi^\sigma}(h^i_t))
    \prob^{\hat{g}^i, \sigma^{-i},\psi^\sigma} (s^{-i}_t, x_t \mid p^i_t, c_t)
    \label{eq:thmclosedness_pf3}
\end{align}
where $\gamma^{\psi^\sigma}(h^i_t) = \psi^\sigma_t(\psi^\sigma_{t-1}, \cdots)$ is the composition of $\psi^\sigma$ from $1$ to $t$.
Then, equation \eqref{eq:thmclosedness_pf1} is true because of \eqref{eq:thmclosedness_pf2} and \eqref{eq:thmclosedness_pf3}.
Consequently, $(S^i_t,\Pi^{\psi^\sigma}_t), t \in \mathcal T$ is an information state for the best response problem for agent $i$ and the assertion of Theorem \ref{thm:closedness} is true.
\end{proof}

As a result of Theorem \ref{thm:closedness}, a definition of SIB BNE equivalent to Definition \ref{def:SIBBNE_original} is the following
\begin{definition}[Equivalent definition of SIB BNE]
\label{def:sibbne}
Consider a SIB strategy profile $\sigma^* = (\sigma^{*1}, \sigma^{*2}, \ldots, \sigma^{*n})$ and its corresponding consistent update rule $\psi^{\sigma^*}$.
The SIB strategy profile $\sigma^*$ is a SIB BNE if for all $i \in \mathcal N$,
\begin{align}
    \mathbb{E}^{\sigma^{i},\sigma^{*-i}, \psi^{\sigma^*}}\{U^i(X_{1:T},A_{1:T})\}
    \leq
    \mathbb{E}^{\sigma^*, \psi^{\sigma^*}}\{U^i(X_{1:T},A_{1:T})\}
\end{align}
for all $\sigma^i \in \Lambda^i$ where $\Lambda^i$ is the set of SIB strategy profiles of agent $i$.
\end{definition}

A consequence of Lemmas \ref{thm:POMDP_new_n}-\ref{lemma:utility_new_n} and Theorem \ref{thm:closedness} is the following. Consider a SIB strategy profile $\sigma$, the corresponding update rule $\psi^\sigma$ along with the consistent CIB belief system $\Pi^{\psi^\sigma}_{1:T}$; if agents $-i$ play according to $\sigma^{-i}$, then the best response of agent $i$ could be determined by the dynamic program
\begin{align}
    \breve V_{T+1}^i(\cdot, \cdot) = 0 \text{ for all }i
\end{align}
\begin{align}
\breve V^i_t(\pi^{\psi^\sigma}_t, s^i_t) = 
& \max_{\tilde \sigma^i_t \in \Lambda^i_t}\ee^{\tilde \sigma^i_t, \sigma^{-i}_t, \psi^\sigma}\{u^i_t(X_t, A_t) + \breve V^i_{t+1}(\psi^{\sigma}_{t+1}(\pi^{\psi^\sigma}_t, Z_{t+1}), S^i_{t+1})\mid s^i_t\}, 
\notag\\
& \,
\forall \pi^{\psi^\sigma}_t \in \Delta(\mathcal X_t \times \mathcal S_t)^N,  \forall s^i_t \in \mathcal S^i_s, t \in \mathcal T
\end{align}
where $\Lambda^i_t$ is the set of SIB strategies of agent $i$ at time $t$.

\subsection{Sequential decomposition}
\label{subsec:dynamicprogram}

Given a set of value functions 
$V_{t+1} = \{V^i_{t+1}: \mathbf\Pi_{t+1} \times \mathcal S^i_{t+1} \rightarrow \mathbb{R}, i \in \mathcal N\}$, a SIB strategy profile $\sigma$, the corresponding update rule $\psi^{\sigma}_{t+1}$ defined by $\eqref{eq:psi_t}$, and the consistent CIB belief $\pi^{\psi^\sigma}_t$,
define the stage-game $G_t(V_{t+1},\pi^{\psi^\sigma}_t)$ as follows.

(i) There are $N$ agents.
(ii) The system state is $X_t$.
(iii) Each agent $i$ observes private information $S^i_t$ and common information $\pi^{\psi^\sigma}_t$.
(iv) Agent $i$'s belief about the state $X_t$ and other agents' private information $S^{-i}_t$ is given by $\pi^{\psi^\sigma,i}_t(x_t, s^{-i}_t)$, that is,
    \begin{align}
        \pi^{\psi^\sigma,i}_t(x_t, s^{-i}_t) \in \Delta(\mathcal X_t \times \mathcal S^{-i}_t).
    \end{align}
(v) Each agent $i$ selects action $A^i_t$ based on his available information; let $\hat \sigma^i_t$ denote agent $i$'s strategy for this stage-game; then,
    \begin{align}
        \prob^{\hat \sigma_t, \psi^\sigma}(A^i_t = a^i_t \mid s^i_t, \pi^{\psi^\sigma}_t) =  \hat\sigma^i_t(a^i_t\mid s^i_t, \pi^{\psi^\sigma}_t).
    \end{align}
(vi) Each agent $i$ has utility
    \begin{align}
        U^i_{G_t(V_{t+1},\pi^{\psi^\sigma}_t)} = u^i_t(X_t, A_t) + V^i_{t+1}(\psi^{\sigma}_{t+1}(\pi^{\psi^\sigma}_t, Z_{t+1}), S^i_{t+1})
        \label{eq:stagegame_utility}
    \end{align}
    where 
    $(Z_{t+1},S^i_{t+1})$ conditioned on $(X_t, S_t, A_t)$ follows the conditional probability $\sum_{x_{t+1}, s^{-i}_{t+1}}\prob(z_{t+1}, x_{t+1}, s_{t+1} \mid x_t, s_t, a_t)$ and the conditional probability $\prob(z_{t+1}, x_{t+1}, s_{t+1} \mid x_t, s_t, a_t)$ is given by
    \begin{align}
    & \prob(z_{t+1}, x_{t+1}, s_{t+1} \mid x_t, s_t, a_t)
    \notag\\
    =& \sum_{y_{t+1}} \mathbb{P}\{x_{t+1}\mid x_t,a_t\}
    \mathbb{P}\{z_{t+1},y_{t+1}\mid x_{t+1},a_t\}
    \notag\\
    & \hspace{1cm}\left(\prod_{j}\mathbbm{1}\{s_{t+1}^j = \phi_{t+1}^j(s_t^j,z_{t+1},y_{t+1}^j,a_t^j)\}\right)
    \label{eq:update_condprob}
    \end{align}

(vii) Given a strategy profile $\hat\sigma_t$ for the stage-game, the expected utility of each player $i$ is given by

    \begin{align}
        &\ee^{\hat\sigma_t, \psi^\sigma} [ U^i_{G_t(V_{t+1},\pi^{\psi^\sigma}_t)} \mid  s^i_t ]
        \notag\\
        = &\sum_{x_t, s_t^{-i}, a_t, z_{t+1}, x_{t+1}, s_{t+1}} 
       \hspace{-1cm} \pi^{\psi^\sigma,i}_t(x_t, s_t^{-i}) \prod_j \hat\sigma^j_t(a^i_t\mid s^i_t, \pi^{\psi^\sigma}_t ) \prob(z_{t+1}, x_{t+1}, s_{t+1} \mid  x_t, s_t, a_t) 
       \notag\\
       &\hspace{2cm} (u^i_t(x_t, a_t) + V^i_{t+1}(\psi^\sigma_{t+1}(\pi_t^{\psi^\sigma}, z_{t+1}), s^i_{t+1}))
        \label{eq:stagegame_exputil}
    \end{align}

Note that all the random variables of the stage-game $G_t(V_{t+1},\pi^{\psi^\sigma}_t)$ may not necessarily be the same as their counterparts in the original dynamic game since each agent $i$ is allowed to choose an arbitrary SIB strategy $\hat \sigma^i_t$ which may be different from $\sigma^i_t$ specified by the SIB strategy profile $\sigma$. The stage-game random variables will coincide with their counterparts in the original game if all agents follow $\sigma$.

\begin{theorem}[Sequential decomposition]
\label{thm:sequential_decomposition}

Consider a SIB strategy profile $\sigma = \{\sigma_t, t \in \mathcal T\}$  and the corresponding update rule $\psi^\sigma = \{\psi^\sigma_t, t \in \mathcal T\}$ defined by \eqref{eq:psi_it}-\eqref{eq:psi_t}.
Define 
\begin{align}
    V_{T+1}^i(\cdot, \cdot) = 0 \text{ for all }i
    \label{eq:dp_last}
\end{align}
\begin{align}
& V^i_t(\pi^{\psi^\sigma_t}, s^i_t) =  \ee^{\sigma_t, \psi^\sigma} [ U^i_{G_t(V_{t+1}, \pi^{\psi^\sigma}_t)}\mid  s^i_t ]
\label{eq:dp_valupdate}
\end{align}
where the right hand side of \eqref{eq:dp_valupdate} is given by \eqref{eq:stagegame_exputil}. If for all $t \in \mathcal T$, there is a SIB strategy profile $\hat\sigma_t$ such that
$\hat\sigma_t$ is a BNE of the stage-game $G_t(V_{t+1}, \pi^{\psi^\sigma}_t)$, that is,

\begin{align}
& \ee^{\hat \sigma^i_t, \hat\sigma^{-i}_t, \psi^\sigma}[U^i_{G_t(V_{t+1},\pi^{\psi^\sigma}_t)}\mid  s^i_t]
 = \max_{\tilde \sigma^i_t \in \Lambda^i_t}\ee^{\tilde \sigma^i_t, \hat\sigma^{-i}_t, \psi^\sigma}[U^i_{G_t(V_{t+1},\pi^{\psi^\sigma}_t)}\mid  s^i_t]
 \label{eq:dp_bne_max}
\end{align}
for all $i \in \mathcal N$ where $\Lambda^i_t$ is the set of SIB strategies of agent $i$ at time $t$, and

\begin{align}
\hat\sigma_t = \sigma_t,
\label{eq:dp_bne_eq}
\end{align}
then the SIB strategy profile $\sigma$ is a SIB-BNE of the original dynamic game.
\end{theorem}

\begin{proof}

Suppose that for all $t\in\mathcal T$ there is a SIB strategy profile $\hat \sigma_t = (\hat\sigma^1_t, \hat\sigma^2_t,\ldots,\hat \sigma^N_t)$ that is a BNE of the stage game $G_t(V_{t+1}, \pi^{\psi^\sigma}_t)$. Then for all $\pi^{\psi^\sigma}_t \in \Delta(\mathcal X_t \times \mathcal S_t)^N,  s^i_t \in \mathcal S^i_s$
\begin{align}
& \ee^{\hat \sigma^i_t, \hat\sigma^{-i}_t, \psi^\sigma}[U^i_{G_t(V_{t+1},\pi^{\psi^\sigma}_t)}\mid  s^i_t]
\notag\\
= &
\max_{\tilde \sigma^i_t \in \Lambda^i_t}\ee^{\tilde \sigma^i_t, \hat\sigma^{-i}_t,\psi^\sigma}[u^i_t(X_t, A_t) + V^i_{t+1}(\psi^{\sigma}_{t+1}(\pi^{\psi^\sigma}_t, Z_{t+1}), S^i_{t+1})\mid  s^i_t].
\label{eq:dp_single}
\end{align}
Equation \eqref{eq:dp_single} holds for all $t \in \mathcal T$ with $V^i_{T+1}(\cdot, \cdot) = 0$ and for all $i\in\mathcal N$.
When $\hat \sigma_t = \sigma_t$ for all $t\in\mathcal T$, Equation \eqref{eq:dp_single} gives, for all $\pi^{\psi^\sigma}_t \in \Delta(\mathcal X_t \times \mathcal S_t)^N,  s^i_t \in \mathcal S^i_s$, 
\begin{align}
V^i_t(\pi^{\psi^\sigma}_t, s^i_t) = 
& \ee^{\sigma^i_t, \sigma^{-i}_t,\psi^\sigma}[U^i_{G_t(V_{t+1},\pi^{\psi^\sigma}_t)}\mid  s^i_t]
\notag\\
= &
\max_{\tilde\sigma^i_t \in \Lambda^i_t}\ee^{\tilde \sigma^i_t, \sigma^{-i}_t,\psi^\sigma}[u^i_t(X_t, A_t) + V^i_{t+1}(\psi^{\sigma}_{t+1}(\pi^{\psi^\sigma}_t, Z_{t+1}), S^i_{t+1})\mid  s^i_t]
\label{eq:thmsq_pf1}
\end{align}
for all $i\in\mathcal N$. 

By induction, \eqref{eq:thmsq_pf1}, and the fact that the update rule $\psi^\sigma$ is consistent with $\sigma$ we have, for all $i\in\mathcal N$ and $t\in\mathcal T$,
\begin{align}
\ee^{\tilde\sigma^i_{t:T}, \sigma^{-i}_{t:T},\psi^\sigma}[
\sum_{\tau = t}^T
u^i_\tau(X_\tau, A_\tau)\mid  s^i_\tau]
\leq 
\ee^{\sigma^i_{t:T}, \sigma^{-i}_{t:T},\psi^\sigma}[
\sum_{\tau = t}^T
u^i_\tau(X_\tau, A_\tau)\mid  s^i_\tau]
\label{eq:thmsq_pf2}
\end{align}
Then \eqref{eq:thmsq_pf2} at time $t=1$ gives 
\begin{align}
    \mathbb{E}^{\tilde\sigma^{i},\sigma^{-i}, \psi^{\sigma}}\{U^i(X_{1:T},A_{1:T})\}
    \leq
    \mathbb{E}^{\sigma, \psi^{\sigma}}\{U^i(X_{1:T},A_{1:T})\}
\end{align}
for all $\tilde \sigma^i \in \Lambda^i$ for all $i\in\mathcal N$.
Therefore, the strategy profile $\sigma$ is a SIB-BNE of the original dynamic game (sf. Definition \ref{def:sibbne}).

\end{proof}

\begin{remark}
\label{remark:stage-game}
Note that even when the stage-game $G_t(V_{t+1},\pi^{\psi^\sigma}_t)$ has a BNE $\hat\sigma_t$, it is possible that $\hat\sigma_t \neq \sigma_t$. Thus, the existence of BNE for every stage-game $G_t(V_{t+1},\pi^{\psi^\sigma}_t)$ is not sufficient to establish the existence of BNE for the original dynamic game.
\end{remark}

\begin{remark}
In the model of \cite{tang2022dynamic} when each team consists of one agent, a SIB BNE coincides with a SPCIB BNE introduced in \cite{tang2022dynamic} with an appropriate mapping of the information state as discussed in Remark \ref{remark:condindepend}.
\end{remark}

\begin{remark}
\label{remark:no_dp_solution}
 There may not be a solution for the set of value functions in the sequential decomposition equations described by \eqref{eq:dp_last}-\eqref{eq:dp_bne_eq} for all $i\in\mathcal N$ and for all $t\in\mathcal T$.
\end{remark}


\begin{remark}
In Definition \ref{def:CIB_update_n}, \eqref{eq:update-2_n} could be defined differently, and different \eqref{eq:update-2_n} would lead to different choices of $\psi$. And for any choice of \eqref{eq:update-2_n}, the claim of Theorem \ref{thm:sequential_decomposition} will still hold.
\end{remark}

\begin{remark}
\label{remark:valuefunction_discont}
The value functions of the sequential decomposition equations defined by Theorem \ref{thm:sequential_decomposition} (Eqs. \eqref{eq:dp_last}-\eqref{eq:dp_bne_eq} for all $i\in\mathcal N, t\in\mathcal T$) may not be continuous in the CIB belief $\Pi^{\psi^\sigma}_t$.
\end{remark}

%% file: example.tex
\section{An illustrative example (Step 4)}
\label{sec:example}

In Section \ref{sec:sequential_decomposition} we argued (cf. Remark \ref{remark:no_dp_solution}) that the sequential decomposition equations defined by \eqref{eq:dp_last}-\eqref{eq:dp_bne_eq} for all $i\in\mathcal N, t\in\mathcal T$ may not have a solution, and that the value functions defined by \eqref{eq:dp_last}-\eqref{eq:dp_bne_eq} may not be continuous in the CIB belief $\Pi^{\psi^\sigma}_t$ (cf. Remark \ref{remark:valuefunction_discont}). In this section we present an example that illustrates/highlights the above remarks. In the example, a two-stage stochastic dynamic game, the agents' utilities depend on a parameter $c$. We show that: (i) the value functions of the corresponding sequential decomposition equations are not continuous in the CIB belief $\Pi^{\psi^\sigma}_t$; (ii) for certain values of $c$ a SIB-BNE exists.

\subsection{Model}

We consider the following two-stage stochastic dynamic game. There are two players/agents, Alice and Bob. At stage one, $t = 1$, the system's state $X_1$ is distributed on $\{-1, 1\}$ with $\mu_0(-1) = \prob(X_1 = -1) = 0.5$ and $\mu_1(1) = \prob(X_1 = 1)=0.5$. Alice observes perfectly $X_1$, i.e., $Y^{Alice}_1=X_1$, and takes action $A^{Alice}_1 \in 
\{-1, 1\}$;  $A^{Alice}_1$ is not observable by Bob and $Y^{Bob}_1=\emptyset$. Bob does not act at $t=1$.
At stage $2$, $t=2$, the system state is $X_2 = X_1 A^{Alice}_1$. 
Alice and Bob have a common observation $Z_2 = X_2A^{Alice}_1W_1 = X_1W_1$, where $W_1 \in \{-1, 1\}$ and $\prob(Z=i \mid  X_1 = i) = 1- p =  0.8, i \in \{-1, 1\}$, and there are no private observations, i.e., $Y^{Alice}_2=Y^{Bob}_2=\emptyset$.
Here $p = 0.2 = \prob(W_1 = -1)$.
Bob acts at $t=2$. Alice does not act at $t=2$.
Bob's action $A^{Bob}_2 \in \{-1, 1\}$. Alice's payoffs at $t=1$ and $t=2$ are 

\begin{align}
    u^{Alice}_1(X_1, A_1) = &\left\{
    \begin{array}{ll}
         c& \text{ if }A^{Alice}_1 = 1 \\
         0& \text{ if }A^{Alice}_1 = -1 
    \end{array}
    \right.
\end{align}
and
\begin{align}
    u^{Alice}_2(X_2, A_2) = &\left\{
    \begin{array}{ll}
         2& \text{ if }X_2 = 1, A^{Bob}_2=1\\
         1& \text{ if }X_2 = -1, A^{Bob}_2=-1\\
         0& \text{ otherwise }
    \end{array}
    \right.
\end{align}
respectively. Bob's payoffs are $u^{Bob}_t(X_t, A_t) = -u^{Alice}_t(X_t, A_t), t= 1, 2$.

The game's information structure is
\begin{align}
    H^{Alice}_1 = & \{X_1\} \\
    H^{Alice}_2 = & \{X_1, A^{Alice}_1, X_2, Z_2\} \\
    H^{Bob}_1 = & \emptyset \\
    H^{Bob}_2 = & \{Z_2\}
\end{align}
where $H^{Alice}_t, H^{Bob}_t, t=1,2$, describe the information available to Alice and Bob, respectively, at stages $1$ and $2$.

This example has the same dynamics and utility functions as Example 3 in \cite{tang2022dynamic}, but Bob doesn't observe Alice's action as in \cite[Example 3]{tang2022dynamic}.

\subsection{Sequential decomposition}

Since Alice perfectly observes the state at both times, i.e., $Y^{Alice}_1 = X_1$ and $Y^{Alice}_2 = X_2$, and Bob doesn't have private information, $S^{Alice}_1 = X_1, S^{Bob}_1 = \emptyset$ are sufficient private information for Alice and Bob at stage $t=1$, respectively, and $S^{Alice}_2 = X_2, S^{Bob}_2 = \emptyset$ are sufficient private information for Alice and Bob, respectively, at stage $t=2$ according to Definition \ref{def:sufficient}.

Suppose $\sigma = (\sigma_1, \sigma_2) = (\sigma^{Alice}_1, \sigma^{Bob}_2)$ is a SIB strategy and $\psi^\sigma$ is the corresponding update rule. Here $\sigma$ is an equilibrium strategy candidate which serves as the strategy prediction for Alice and Bob.
Note that
$\Pi_1^{\psi^\sigma, Alice}(x_1) = \mu_0(x_1)$ and $\Pi_1^{\psi^\sigma, Bob}(x_1) = \mu_0(x_1)$ for all $x_1 \in \mathcal X_1$.

To get a BNE using the sequential decomposition of Theorem \ref{thm:sequential_decomposition}, we first consider the stage-game $G_2(0, \pi_2^{\psi^\sigma})$ at time $2$. Since Bob is the only agent who acts at time $2$ and $S^{Bob}_2 = \emptyset$, any BNE $\sigma_2$ of $G_2(0, \pi_2^{\psi^\sigma})$ must satisfy
\begin{align}
     \hat\sigma_2^{Bob} = & \argmax_{\tilde\sigma^{Bob}_2} \ee^{\tilde\sigma^{Bob}_2, \psi^\sigma} [ u^{Bob}_2(X_2, A_2)]
     \notag\\
     = & \argmax_{\tilde\sigma^{Bob}_2}\Big(
     -2 \prob^{\tilde\sigma^{Bob}_2, ^{\psi^\sigma}}(X_2 = A^{Bob}_2 = 1)
     -
     \prob^{\tilde\sigma^{Bob}_2, \psi^\sigma}(X_2 = A^{Bob}_2 = -1)
     \Big)
     \notag\\
     = & \argmax_{\tilde\sigma^{Bob}_2}\Big(
     -2 \pi^{\psi^\sigma,Bob}_2(1)\tilde\sigma^{\psi^\sigma,Bob}_2(1 \mid  \pi_2^{\psi^\sigma})
     \notag\\
     & \hspace{2cm} - (1 - \pi^{\psi^\sigma,Bob}_2(1))(1 - \tilde\sigma^{\psi^\sigma,Bob}_2(1 \mid  \pi_2^{\psi^\sigma}))
     \Big)
     \label{eq:example_bne_time2}
\end{align}
From \eqref{eq:example_bne_time2} we conclude that 
one of the equilibrium SIB strategies is given by
\begin{align*}
    & \sigma^{Bob}_2(\pi_2^{\psi^\sigma}) = 1 \text{, if } \pi^{\psi^\sigma,Bob}_2(1) \leq 1/3, 
    \\
    & \sigma^{Bob}_2(\pi_2^{\psi^\sigma}) = 0  \text{, if } \pi^{\psi^\sigma,Bob}_2(1) > 1/3,
\end{align*}
or equivalently
\begin{align}
    & \sigma^{Bob}_2(\pi_2^{\psi^\sigma}) = \mathds{1}(\pi^{\psi^\sigma,Bob}_2(1) \leq 1/3) 
    \label{eq:bob_bestresponse}
\end{align}
Note that $\sigma^{Bob}_2(\pi_2^{\psi^\sigma})$ can take any value in $[0, 1]$ if $\pi^{\psi^\sigma,Bob}_2(1) = 1/3$ and $\sigma_2$ is still a BNE of the stage-game.

Alice's sufficient private information at time $2$ is $S^{Alice}_2 = X_2$.
With the stage-game equilibrium SIB strategy $\sigma^{Bob}_2(\pi_2)$ given by \eqref{eq:bob_bestresponse}, the value function for Alice at $t=2$ is then given, according to \eqref{eq:dp_valupdate}, by
\begin{align}
     V^{Alice}_2(\pi_2^{\psi^\sigma}, x_2) = &  \ee^{\sigma_2, \psi^\sigma} [ u^{Alice}_2(X_2, A_2) \mid  x_2]
     \notag\\
     = & \left\{
     \begin{array}{ll}
           2 \mathds{1}(\pi^{\psi^\sigma,Bob}_2(1) \leq 1/3)  & \text{if } x_2 = 1
          \\
          1 - \mathds{1}(\pi^{\psi^\sigma,Bob}_2(1) \leq 1/3)  & \text{if } x_2 = - 1
     \end{array}\right.
\end{align}

Given the above value functions at time $t=2$, we now consider the stage-game $G_1(V_2, \pi^{\psi^\sigma}_1)$ at time $t=1$.
The utility for the stage-game for Alice is given as follows.
\begin{align}
    U^{Alice}_{G_1(V_2, \pi^{\psi^\sigma}_1)}
    = u^{Alice}_1(X_1, A_1) + V^{Alice}_2(\psi^\sigma_2(\pi_1, Z), X_2)
\end{align}
If Alice uses the SIB strategy $\tilde \sigma^{Alice}_1$, the expected utility of the stage-game can be calculated for $X_1=-1$ and $X_1=1$, according to \eqref{eq:stagegame_exputil}, by
\begin{align}
    & \ee^{\tilde\sigma^{Alice}_1, \psi^\sigma}[
    U^{Alice}_{G_1(V_2, \pi^{\psi^\sigma}_1)} \mid  X_1 = -1] 
    \notag\\
    = & c\tilde\sigma^{Alice}_1(1\mid -1) + \ee^{\tilde\sigma^{Alice}_1, \psi^\sigma}[V^A_2(\psi_{2}^\sigma(\pi_1^{\psi^\sigma}, X_1 W_1), X_1 A^{Alice}_1) \mid  X_1 = -1]
    \notag\\
    = & (1 + c)(1 - \tilde\alpha_1) + (3\tilde\alpha_1 - 1)((1-p)
    \mathds{1}(q_{-1} \leq 1/3) 
     + p\mathds{1}(q_{1} \leq 1/3) )
    \notag\\
    =: & r^A_{-1}(\tilde\alpha_1, q)
    \\
    & \ee^{\tilde\sigma^{Alice}_1, \psi^\sigma}[
    U^{Alice}_{G_1(V_2, \pi^{\psi^\sigma}_1)} \mid  X_1 = 1]  
    \notag\\
    = & c\tilde\sigma^{Alice}_1(1\mid 1) + \ee^{\tilde\sigma^{Alice}_1, \psi^\sigma}[V^A_2(\psi^\sigma_{2}(\pi_1^{\psi^\sigma}, X_1 W_1), X_1 A^{Alice}_1) \mid  X_1 = 1]
    \notag\\
    = & 1 + (c - 1)\tilde\alpha_2 + (3\tilde\alpha_2 - 1)((1-p)\mathds{1}(q_{1} \leq 1/3)  + p\mathds{1}(q_{-1} \leq 1/3) )
    \notag\\
    =: & r^A_{1}(\tilde\alpha_2, q)
\end{align}
where $q = (q_{-1}, q_1)$, $q_{-1} = \psi_{2}^{\sigma, Bob}(\pi_1^{\psi^\sigma}, -1)(1)$ and $q_1= \psi_{2}^{\sigma, Bob}(\pi_1^{\psi^\sigma}, 1)(1)$ are the CIB beliefs $\pi^{\psi^\sigma,Bob}_2(1)$ of $\{X_2 = 1\}$ when $Z=-1$ and $Z=1$, respectively, and $\tilde\alpha = (\tilde\alpha_1, \tilde\alpha_2)$, $\tilde\alpha_1 = \tilde\sigma^{Alice}_1(-1\mid -1), \tilde\alpha_2 = \tilde\sigma^{Alice}_1(1\mid 1)$ represents Alice's SIB strategy $\tilde\sigma^{Alice}_1$.

Note that from Bayes' rule in Definition \ref{def:CIB_update_n}, under the SIB strategy $\sigma^{Alice}_1$, represented by $\alpha_1 = \sigma^{Alice}_1(-1\mid -1)$ and $\alpha_2 = \sigma^{Alice}_1(1\mid 1)$, we have
\begin{align}
    & q_{-1} = \psi_{2}^{\psi^\sigma,Bob}(\pi_1^{\psi^\sigma}, -1)(1) = \frac{\prob^{\alpha}(X_2 = 1, Z= -1)}{\prob^{\alpha}(Z = -1)} = \alpha_2 p + \alpha_1 (1-p)
    \label{eq:qm1}
    \\
    & q_1 = \psi_{2}^{\psi^\sigma,Bob}(\pi_1^{\psi^\sigma}, 1)(1) = \frac{\prob^{\alpha}(X_2 = 1, Z= 1)}{\prob^{\alpha}(Z = 1)} = \alpha_2 (1 - p) + \alpha_1 p 
    \label{eq:q1}
\end{align}
Therefore, a SIB strategy $\hat\sigma^{Alice}_1$, represented by $\hat\alpha_1 = \hat\sigma^{Alice}_1(-1\mid -1)$ and $\hat\alpha_2 = \hat\sigma^{Alice}_1(1\mid 1)$, is a BNE of the stage-game $G_1(V_2, \pi^{\psi^\sigma}_1)$ at time $t=1$ if 
\begin{align}
    & \hat\alpha_1 \in 
    \argmax_{\tilde\alpha_1} 
    r^A_{-1}(\tilde\alpha_1, ( \alpha_2 p + \alpha_1 (1-p), \alpha_2 (1 - p) + \alpha_1 p ))
    \\
    & \hat\alpha_2 \in 
    \argmax_{\tilde\alpha_2} 
    r^A_{1}(\tilde\alpha_2, ( \alpha_2 p + \alpha_1 (1-p), \alpha_2 (1 - p) + \alpha_1 p ))
\end{align}
Consequently, the SIB strategy $\sigma^{Alice}_1$, represented by $\alpha_1 = \sigma^{Alice}_1(-1\mid -1)$ and $\alpha_2 = \sigma^{Alice}_1(1\mid 1)$ will satisfy the sequential decomposition equations \eqref{eq:dp_bne_max}-\eqref{eq:dp_bne_eq} if
\begin{align}
    & \alpha_1 \in 
    \argmax_{\tilde\alpha_1} 
    r^A_{-1}(\tilde\alpha_1, ( \alpha_2 p + \alpha_1 (1-p), \alpha_2 (1 - p) + \alpha_1 p ))
    \label{eq:br_rm1}
    \\
    & \alpha_2 \in 
    \argmax_{\tilde\alpha_2} 
    r^A_{1}(\tilde\alpha_2, ( \alpha_2 p + \alpha_1 (1-p), \alpha_2 (1 - p) + \alpha_1 p ))
    \label{eq:br_r1}
\end{align}

\begin{remark}
Note that the functions $r^A_{-1}(\tilde\alpha_1, q)$ and $r^A_{1}(\tilde\alpha_2, q)$ are not continuous in $q$. Thus existence of equilibria cannot be established by the standard method relying on the continuity of the utility functions, and there may not no equilibria in the general case.
\end{remark}

\subsection{Existence of SIB-BNE under conditions on the instantaneous utility.}

The stage-game $G_1(V_2, \pi^{\psi^\sigma}_1)$ is a normal-form game with a fixed $\sigma_1$.
According to Remark \ref{remark:stage-game}, a BNE $\hat \sigma$ of $G_1(V_2, \pi^{\psi^\sigma}_1)$ could be different from $\sigma_1$ and the existence of a regular BNE of $G_1(V_2, \pi^{\psi^\sigma}_1)$ is not sufficient to satisfy \eqref{eq:dp_bne_eq} at time $t=1$. In order to apply equilibrium existence results for normal-form games to the sequential decomposition at time $t=1$, we introduce an agent $0$ who picks the $q$-belief $q = (q_{-1}, q_{1})$ so that \eqref{eq:dp_bne_eq} is satisfied.

Formally, we construct an augmented stage-game $\hat G_1$ between Alice and agent $0$. Alice chooses $\tilde\alpha = (\tilde\alpha_1, \tilde\alpha_2)$ and agent $0$ chooses $\tilde q = (\tilde q_{-1}, \tilde q_{1})$.
Alice's utility is
\begin{align}
    r^A_1(\tilde \alpha, \tilde q)
    = & 0.5 r^A_{-1}(\tilde \alpha_1, \tilde q) + 0.5 r^A_{1}(\tilde \alpha_2, \tilde q)
    \notag\\
    = & 0.5c(1 - \tilde \alpha_1 + \tilde \alpha_2) + 0.5(2 - \tilde \alpha_1-\tilde \alpha_2)
     \notag\\
    & + 0.5(3(\tilde \alpha_2 p + \tilde \alpha_1 (1-p))- 1)\mathds{1}(\tilde q_{-1} \leq 1/3)
    \notag\\
    & + 0.5(3 (\tilde \alpha_2 (1 - p) + \tilde \alpha_1 p) - 1)\mathds{1}(\tilde q_{1} \leq 1/3).
     \label{eq:agument_alice}
\end{align}
Agent $0$'s utility is
\begin{align}
r^0_1(\tilde\alpha, \tilde q) = 
    -(\tilde q_{-1}  - \tilde \alpha_2 p - \tilde \alpha_1 (1-p) )^2
     - (\tilde q_{1}  - \tilde \alpha_2 (1 - p) - \tilde \alpha_1 p)^2.
     \label{eq:agument_zero}
\end{align}
Both Alice and agent $0$ are utility maximizers.
The game $\hat G_1$ with utilities \eqref{eq:agument_zero}-\eqref{eq:agument_alice} is a normal-form game with strategies $\tilde \alpha = (\tilde \alpha_1, \tilde \alpha_2)$  $\tilde q = (\tilde q_{-1}, \tilde q_{1})$. 
Since the utility \eqref{eq:agument_zero} of agent $0$ is a quadratic function, any best response by agent $0$ must satisfy $\tilde q_{-1}  = \tilde \alpha_2 p + \tilde \alpha_1 (1-p), \tilde q_{1}  = \tilde \alpha_2 (1 - p) + \tilde \alpha_1 p$.

Note that in the augmented stage-game $\hat G_1$, the utility function $r^A_1(\tilde \alpha, \tilde q)$ is not continuous in $\tilde q$. To show the existence of a Nash equilibrium for $\hat G_1$, we proceed to apply existence results for games with discontinuous utilities in \cite{barelli2013note}. 

Specifically, Proposition 2.4 of \cite{barelli2013note} guarantees the existence of a Nash equilibrium for games satisfying the generalized better reply secure property.
From Definition 2.3 in \cite{barelli2013note}, the stage game is generalized better reply secure if for any $(\bar\alpha, \bar q)$ not an equilibrium, at least one of the followings is true
\begin{itemize}
    \item We can find an $\epsilon > 0$ and a closed correspondence $\phi^0(\tilde\alpha, \tilde q)$ such that 
    \begin{align}
        r^0_1(\tilde\alpha, \phi^0(\tilde\alpha, \tilde q)) \geq r^0_1(\bar\alpha, \bar q) + \epsilon
    \end{align}
    for all $\tilde\alpha_1 \in (\bar\alpha_1 - \epsilon, \bar\alpha_1 + \epsilon)$, $\tilde\alpha_2 \in (\bar\alpha_2 - \epsilon, \bar\alpha_2 + \epsilon)$, $\tilde q_{-1} \in (\bar q_{-1} - \epsilon, \bar q_{-1} + \epsilon)$, $\tilde q_{1} \in (\bar q_{1} - \epsilon, \bar q_{1} + \epsilon)$
    \item We can find an $\epsilon > 0$ and a closed correspondence $\phi^A(\tilde\alpha, \tilde q)$ such that 
    \begin{align}
        r^A_1(\phi^A(\tilde\alpha, \tilde q), \tilde q) \geq r^A_1(\bar\alpha, \bar q) + \epsilon
    \end{align}
    for all $\tilde\alpha_1 \in (\bar\alpha_1 - \epsilon, \bar\alpha_2 + \epsilon)$, $\tilde\alpha_2 \in (\bar\alpha_2 - \epsilon, \bar\alpha_2 + \epsilon)$, $\tilde q_{-1} \in (\bar q_{-1} - \epsilon, \bar q_{-1} + \epsilon)$, $\tilde q_{1} \in (\bar q_{1} - \epsilon, \bar q_{1} + \epsilon)$
\end{itemize}

In Appendix \ref{appendix:example}, we show that when $c > 24$ the augmented stage-game $\hat G_1$ is generalized better reply secure. 
Thus, there exists a Nash equilibrium of the augmented state-game $\hat G_1$ according to \cite[Proposition 2.4]{barelli2013note}.

Consider any Nash equilibrium $(\alpha, q)$ of $\hat G_1$. Since $q$ is a best response to $\alpha$ for agent $0$, from agent $0$'s utility \eqref{eq:agument_zero} we have 
\begin{align}
    & q_{-1} = \alpha_2 p + \alpha_1 (1-p)
    \\
    & q_1  = \alpha_2 (1 - p) + \alpha_1 p 
\end{align}
Furthermore, since $\alpha$ is a best response to $q$ for Alice in $\hat G_1$,
\begin{align}
    \alpha \in &\argmax_{\tilde \alpha } 
    \Big(
    0.5 r^A_{-1}(\tilde \alpha_1, q) + 0.5 r^A_{1}(\tilde \alpha_2, q)
    \Big)
    \notag\\
    = & \argmax_{\tilde \alpha } 
    \Big(
    0.5 r^A_{-1}(\tilde \alpha_1, (\alpha_2 p + \alpha_1 (1-p), \alpha_2 (1 - p) + \alpha_1 p )) 
    \notag\\
    & \hspace{2cm}+ 0.5 r^A_{1}(\tilde \alpha_2, (\alpha_2 p + \alpha_1 (1-p), \alpha_2 (1 - p) + \alpha_1 p ))
    \Big)
    \notag\\
    = & \Big(\argmax_{\tilde \alpha_1 } r^A_{-1}(\tilde \alpha_1, (\alpha_2 p + \alpha_1 (1-p), \alpha_2 (1 - p) + \alpha_1 p )) , 
    \notag\\
    & \hspace{1cm} \argmax_{\tilde \alpha_2 }r^A_{1}(\tilde \alpha_2, (\alpha_2 p + \alpha_1 (1-p), \alpha_2 (1 - p) + \alpha_1 p ))
    \Big)
\end{align}
Therefore, \eqref{eq:br_rm1}-\eqref{eq:br_r1} hold for $\alpha$, and consequently the sequential decomposition requirement \eqref{eq:dp_bne_max}-\eqref{eq:dp_bne_eq} is satisfied at $t=1$ by the SIB strategy $\sigma^{Alice}_1$ represented by $\alpha$, and we establish the existence of a SIB equilibrium based on Theorem \ref{thm:sequential_decomposition}.

%% file: specialcase.tex
\section{The case with no common observations}

\label{sec:specialcase}
We consider the model of Section \ref{sec:model} but we assume that the agents have no common observations, that is, 
\begin{align}
    Z_t = \emptyset \quad \forall t \in \mathcal T .
\end{align}
The system's dynamics, the agents' private observations, the functional form of the agents' strategies, their utilities, and the equilibrium concept (BNE) remain the same as in Section \ref{sec:model}.

Even though the agents have no common observations in this special case, we can still define SIB strategies by Definition \ref{def:SIB_strategy_n}, and construct the consistent CIB belief system according to Definition \ref{def:CIB_update_n} with $Z_t = \emptyset \, \forall t \in \mathcal T$.

Since there is no common observations, for any realization we always have
\begin{align}
    & \sum_{\hat{x}_{t+1},\hat{s}_{t+1}} F_{t}^i(\hat{x}_{t+1},\hat{s}_{t+1},z_{t+1})(\pi^{\psi^\sigma}_t;\sigma^{-i}_t) 
    \notag\\
    = & \sum_{\hat{x}_{t+1},\hat{s}_{t+1}} F_{t}^i(\hat{x}_{t+1},\hat{s}_{t+1})(\pi^{\psi^\sigma}_t;\sigma^{-i}_t) = 1 > 0
\end{align}
Therefore, case (ii) in Definition \ref{def:CIB_update_n} would never happen, and \eqref{eq:cib_bayesrule} can be simplified to
\begin{align}
& \pi_{t+1}^{\psi^{\sigma},i}(x_{t+1},s_{t+1}) 
\notag\\
= & \frac{F_{t}^i(x_{t+1},s_{t+1})(\pi^{\psi^\sigma}_t;\sigma^{-i}_t)}{\sum_{\hat{x}_{t+1},\hat{s}_{t+1}} F_{t}^i(\hat{x}_{t+1},\hat{s}_{t+1})(\pi_t^{\psi^\sigma};\sigma^{-i}_t)}
\notag\\
= & F_{t}^i(x_{t+1},s_{t+1})(\pi^{\psi^\sigma}_t;\sigma^{-i}_t)
\notag\\
= & 
\sum_{y_{t+1},x_t,s_t,a_t} \Bigg[\mathbb{P}\{y_{t+1},x_{t+1}\mid x_t,a_t\}\left(\prod_{j}\mathbbm{1}\{s_{t+1}^j = \phi_{t+1}^j(s_t^j,y_{t+1}^j,a_t^j)\}\right)\nonumber\\
& \hspace{50pt} \left(\frac{1}{\vert A_t^i\vert}\prod_{j\neq i} \sigma^j_t(a^j_t)(\pi^\psi_t,s_t^j) \right) \pi_t^{\psi,i}(x_t,s_t) \Bigg].
\label{eq:cib_bayesrule_nocommon}
\end{align}
Based on \eqref{eq:cib_bayesrule_nocommon} we can write \begin{align}
    &\Pi_{t+1}^{\psi^\sigma,i} = \psi_{t+1}^{\sigma,i}(\Pi_t^{\psi^\sigma}) \quad \forall i \in \mathcal N,
    \label{eq:psi_it_nocommon}
\\
    &\Pi_{t+1}^{\psi^\sigma} = \psi^{\sigma}_{t+1}(\Pi_t^{\psi^\sigma}).
    \label{eq:psi_t_nocommon}
\end{align}

In other words, given a SIB strategy $\sigma$, the update rule $\psi^\sigma$ are deterministic functions given by \eqref{eq:psi_t_nocommon}, and the corresponding consistent CIB belief system $\Pi_{t}^{\psi^\sigma}, t \in \mathcal T$, evolves in a deterministic manner.
Furthermore, since case (ii) in Definition \ref{def:CIB_update_n} never happens without common observations, the update rule $\psi_{t+1}^{\sigma,i}$ given by \eqref{eq:cib_bayesrule_nocommon} becomes exactly the Bayes rule. As a result, the CIB belief $\Pi_{t}^{\psi^\sigma,i}$ becomes a regular PMF given by
\begin{align}
    \Pi_{t}^{\psi^\sigma,i}(x_t, s_t) = \prob^{\tilde g^i, \sigma^{-i}}(x_t, s_t) \quad \forall i \in \mathcal N
\end{align}
where $\tilde g^i$ denotes the uniform strategy (i.e., the strategy that chooses every action $a^i_t \in \mathcal A^i_t$ with equal probability for all $t \in \mathcal T$).

\begin{remark}
    If the $N$ agents have identical utilities, i.e. we have a dynamic team problem, then  $\Pi_{t}^{\psi^\sigma}, t \in \mathcal T$ is similar to the common knowledge that appears in \cite{witsenhausen1973standard} where a dynamic team is analyzed. The common knowledge in \cite{witsenhausen1973standard} is a sequence (over time) of PMFs on the system's history $H_t, t \in \mathcal T$. These PMFs evolve in a deterministic manner, similar to \eqref{eq:cib_bayesrule_nocommon} for $\Pi_{t}^{\psi^\sigma}, t \in \mathcal T$, in the model of this section.
\end{remark}

For this special case with no common observations, Theorem \ref{thm:sequential_decomposition} becomes
\begin{corollary}
\label{cor:no_common}
    Consider a SIB strategy profile $\sigma = \{\sigma_t, t \in \mathcal T\}$  and the corresponding update rule $\psi^\sigma = \{\psi^\sigma_t, t \in \mathcal T\}$ defined by \eqref{eq:psi_it_nocommon}-\eqref{eq:psi_t_nocommon} for the model of this section.
Define 
\begin{align}
    V_{T+1}^i(\cdot, \cdot) = 0 \text{ for all }i
\end{align}
\begin{align}
& V^i_t(\pi^{\psi^\sigma_t}, s^i_t) =  \ee^{\sigma_t, \psi^\sigma} [ U^i_{G_t(V_{t+1}, \pi^{\psi^\sigma_t})}\mid  s^i_t ]
\end{align}
where $U^i_{G_t(V_{t+1},\pi^{\psi^\sigma}_t)} = u^i_t(X_t, A_t) + V^i_{t+1}(\psi^{\sigma}_{t+1}(\pi^{\psi^\sigma}_t), S^i_{t+1})$, and in the conditional expectation $\ee^{\sigma_t, \psi^\sigma}[\cdot]$, the distribution of $(X_t, S_t)$ conditioned on $S^i_t$ is given by $\pi^{\psi^\sigma,i}_t(x_t, s^{-i}_t)$, $A^i_t, i \in \mathcal N$, are generated by $\sigma^i_t(a^i_t\mid s^i_t, \pi^{\psi^\sigma}_t)$, $S^i_{t+1}$ conditioned on $(X_t, S_t, A_t)$ follows the conditional probability $\sum_{x_{t+1}, s^{-i}_{t+1}}\prob(x_{t+1}, s_{t+1} \mid  x_t, s_t, a_t)$ given by
    \begin{align}
    & \prob(x_{t+1}, s_{t+1} \mid  x_t, s_t, a_t)
    \notag\\
    =& \sum_{y_{t+1}} \mathbb{P}\{x_{t+1}\mid x_t,a_t\}
    \mathbb{P}\{y_{t+1}\mid x_{t+1},a_t\}
    \left(\prod_{j}\mathbbm{1}\{s_{t+1}^j = \phi_{t+1}^j(s_t^j,y_{t+1}^j,a_t^j)\}\right).
    \label{eq:update_condprob_common}
    \end{align}
If for all $t \in \mathcal T$, there is a SIB strategy profile $\hat\sigma_t$ such that
$\hat\sigma_t$ is a BNE of the stage-game $G_t(V_{t+1}, \pi^{\psi^\sigma_t})$, that is,
\begin{align}
& \ee^{\hat \sigma^i_t, \hat\sigma^{-i}_t, \psi^\sigma}[U^i_{G_t(V_{t+1},\pi^{\psi^\sigma}_t)}\mid  s^i_t]
 = \max_{\tilde \sigma^i_t \in \Lambda^i_t}\ee^{\tilde \sigma^i_t, \hat\sigma^{-i}_t, \psi^\sigma}[U^i_{G_t(V_{t+1},\pi^{\psi^\sigma}_t)}\mid  s^i_t]
 \label{eq:dp_bne_max_nocommon}
\end{align}
for all $i \in \mathcal N$, and
\begin{align}
\hat\sigma_t = \sigma_t,
\label{eq:dp_bne_eq_nocommon}
\end{align}
then the SIB strategy profile $\sigma$ is a SIB-BNE of the dynamic game without common observations defined in this section.
\end{corollary}

\begin{remark}
    The SIB-BNE strategy profiles $\{\sigma_t, t \in \mathcal T\}$ determined by sequential decomposition in Corollary \ref{cor:no_common}, along with the beliefs $\{\Pi^{\psi^{\sigma}}_t, t \in \mathcal T\}$ are also Perfect Bayesian Equilibria (PBE) \cite{fudenberg1991game}. This is true because $\{\sigma_t, t \in \mathcal T\}$ satisfy sequential rationality (Eq. \eqref{eq:dp_bne_max_nocommon}) and consistency holds because the beliefs $\{\Pi^{\psi^{\sigma}}_t, t \in \mathcal T\}$ are always updated by Bayes rule.
\end{remark}

%% file: appendix.tex
\appendix
\subsection{Sufficient Information}
\label{app:sufficient_information}

We compare conditions (i)-(iii) of Definition \ref{def:sufficient} to the conditions of Definition 2 in \cite{companion}; for ease of readability, we include the definition from \cite{companion} below.
\begin{definition}[Sufficient private information \cite{companion}]
	We say $S_t^i=\zeta_t^i(P_t^i,C_t;g_{1:t-1})$, $i\in\mathcal{N}$, $t\in\mathcal{T}$, is \textit{sufficient private information} for the agents if, 
	\begin{enumerate}[(i)]
		\item it can be updated recursively as \vspace*{-2pt}
		\begin{gather}
		S_t^{i}=\phi_t^i(S_{t-1}^{i},H_t^i\backslash H_{t-1}^i;g_{1:t-1})  \text{ for } t\in\mathcal{T}\backslash\{1\}, \label{eq:sufficientupdate}
		\end{gather}
		\item for any strategy profile $g$ and for all realizations $\{c_t,p_t,p_{t+1},z_{t+1},a_t\}\in\mathcal{C}_t\times\mathcal{P}_t\times\mathcal{P}_{t+1}\times\mathcal{Z}_{t+1}$ of positive probability,
		\begin{align}
		\mathbb{P}^{g_{1:t}}\left\{\hspace*{-2pt}s_{t+1}\hspace*{-1pt},\hspace*{-1pt}z_{t+1}\hspace*{-1pt}\mid p_t\hspace*{-1pt},\hspace*{-1pt}c_t\hspace*{-1pt},\hspace*{-1pt}a_t\hspace*{-2pt}\right\}\hspace*{-3pt}=\hspace*{-2pt}\mathbb{P}^{g_{1:t}}\hspace*{-1pt}\left\{\hspace*{-2pt}s_{t+1}\hspace*{-1pt},\hspace*{-1pt}z_{t+1}\hspace*{-1pt}\mid s_t\hspace*{-1pt},\hspace*{-1pt}c_t\hspace*{-1pt},\hspace*{-1pt}a_t\hspace*{-2pt}\right\}\hspace*{-1pt},\hspace*{-4pt}\label{eq:sufficientdynamic}
		\end{align}
		where $s_{\tau}^{1:N}=\zeta_{\tau}^{1:N}(p_{\tau}^{1:N},c_{\tau};g_{1\hspace*{-1pt}:\tau-1})$ for $\tau\in\mathcal{T}$;	
		\item for every strategy profile $\tilde{g}$ of the form		
		$\tilde{g}\hspace*{-2pt}:=\hspace*{-2pt}\{\hspace*{-1pt}\tilde{g}^{i}_t\hspace*{-1pt}:\hspace*{-1pt}\mathcal{S}_t^i\times \mathcal{C}_t\rightarrow \Delta(\mathcal{A}_t^i), i\hspace*{-2pt}\in\hspace*{-2pt}\mathcal{N}\hspace*{-1pt},\hspace*{-1pt} t\hspace*{-2pt}\in\hspace*{-2pt}\mathcal{T}\}$ and $a_t\hspace*{-2pt}\in\hspace*{-2pt}\mathcal{A}_t$, $t\hspace*{-2pt}\in\hspace*{-2pt}\mathcal{T}$;
		\begin{align} 		
		\mathbb{E}^{\tilde{g}_{1:t-1}\hspace*{-2pt}}\left\{\hspace*{-2pt}u_t^i(\hspace*{-1pt}X_t\hspace*{-1pt},\hspace*{-1pt}A_t\hspace*{-1pt})\hspace*{-1pt}\mid c_t\hspace*{-1pt},\hspace*{-1pt}p_t^i\hspace*{-1pt},\hspace*{-1pt}a_t\hspace*{-2pt}\right\}\hspace*{-3pt}=\hspace*{-2pt}\mathbb{E}^{\tilde{g}_{1:t-1}\hspace*{-2pt}}\left\{\hspace*{-2pt}u_t^i(\hspace*{-1pt}X_t\hspace*{-1pt},\hspace*{-1pt}A_t\hspace*{-1pt})\hspace*{-1pt}\mid c_t\hspace*{-1pt},\hspace*{-1pt}s_t^{i}\hspace*{-1pt},\hspace*{-1pt}a_t\hspace*{-2pt}\right\}\hspace*{-2pt},\hspace*{-5pt}\label{eq:payoff-relevant2}
		\end{align} 
		for all realizations $\{\hspace*{-1pt}c_{t}\hspace*{-1pt},\hspace*{-1pt}p_{t}^i\}\hspace*{-3pt}\in\hspace*{-2pt}\mathcal{C}_{t}\hspace*{-1pt}\times\hspace*{-1pt}\mathcal{P}_{t}^i$ of positive probability where $s_{\tau}^{1:N}\hspace*{-3pt}=\hspace*{-2pt}\zeta_{\tau}^{1:N}\hspace*{-2pt}(p_{\tau}^{1:N}\hspace*{-1pt},\hspace*{-1pt}c_{\tau};\hspace*{-1pt}\tilde{g}_{1\hspace*{-1pt}:\tau-1}\hspace*{-1pt})$ for $\tau\in\mathcal{T}$;\vspace{5pt}	\item given an arbitrary strategy profile $\tilde{g}$ of the form $\tilde{g}\hspace*{-1pt}:=\hspace*{-1pt}\{\tilde{g}^{i}_t:\mathcal{S}_t^i\hspace*{-1pt}\times \hspace*{-1pt}\mathcal{C}_t\rightarrow \Delta(\mathcal{A}_t^i), i\hspace*{-2pt}\in\hspace*{-2pt}\mathcal{N}, t\hspace*{-2pt}\in\hspace*{-2pt}\mathcal{T}\}$, $i\hspace*{-2pt}\in\hspace*{-2pt}\mathcal{N}$, and $t\hspace*{-2pt}\in\hspace*{-2pt}\mathcal{T}$,
		\begin{align}
		\mathbb{P}^{\tilde{g}_{1:t-1}}\hspace*{-2pt}\left\{\hspace*{-2pt}s_t^{-i}\hspace*{-1pt}\mid p_t^i\hspace*{-1pt},\hspace*{-1pt}c_t\hspace*{-2pt}\right\}\hspace*{-3pt}=\hspace*{-2pt}\mathbb{P}^{\tilde{g}_{1:t-1}}\hspace*{-2pt}\left\{\hspace*{-1pt}s_t^{-i}\hspace*{-1pt}\mid s_t^i\hspace*{-1pt},\hspace*{-1pt}c_t\hspace*{-2pt}\right\}\hspace*{-1pt},\hspace*{-4pt}\label{eq:sufficientinfo}
		\end{align}
		for all realizations $\{c_{t}\hspace*{-1pt},\hspace*{-1pt}p_{t}^i\}\hspace*{-2pt}\in\hspace*{-2pt}\mathcal{C}_{t}\hspace*{-2pt}\times\hspace*{-2pt}\mathcal{P}_{t}^i$ of positive probability where $s_{\tau}^{1:N}\hspace*{-3pt}=\hspace*{-2pt}\zeta_{\tau}^{1:N}\hspace*{-2pt}(p_{\tau}^{1:N}\hspace*{-1pt},\hspace*{-1pt}c_{\tau};\hspace*{-1pt}\tilde{g}_{1\hspace*{-1pt}:\tau-1}\hspace*{-1pt})$ for $\tau\in\mathcal{T}$.		
	\end{enumerate}
	\label{def:sufficient-part1}
\end{definition}
Condition (i) of Definition \ref{def:sufficient} appears in the definition of $S^i_t$ in Definition \ref{def:sufficient-part1}, and condition (ii) of Definition \ref{def:sufficient} on recursive update is the same as condition (i) in Definition \ref{def:sufficient-part1}.
Condition (iii) of Definition \ref{def:sufficient} directly leads to (iii) and (iv) of Definition \ref{def:sufficient-part1}; 
the utility $u_t^i(X_t, A_t)$ in 
condition (iii) and the random variable $s^{-i}_t$ in condition $(iv)$ of Definition \ref{def:sufficient-part1} are functions of $(x_t, s_t)$ whose distribution conditioned on $(p^i_t, c_t)$ is the same as conditioned on $(s^i_t, c_t)$ under condition (iii) of Definition \ref{def:sufficient}.

However, condition (ii) of Definition \ref{def:sufficient-part1} may not hold for sufficient private information satisfying Definition \ref{def:sufficient}. 
Consider the following example.
Suppose $X_1 = Y^1_1 \text{ XOR } Y^2_1$, and $Y^1_1, Y^2_1$ takes values in $\{0, 1\}$ with equal probability. $Z_1 = \emptyset$ and $Z_2 = X_1$. Then $S^1_1=S^2_1=\emptyset$ satisfies Definition \ref{def:sufficient} because $\prob(x_1, s^{-i}_1\mid p^i_1, c_1) = \prob(x_1\mid y^i_1) = 0.5 = \prob(x_1, s^{-i}_1\mid s^i_1, c_1)$. However, they don't satisfy condition (ii) of Definition \ref{def:sufficient-part1} because $\prob(z_2\mid p_1, c_1, a_1) = \prob(x_1\mid y^1_1, y^2_1) = \mathds{1}(x_1 = y^1_1 \text{ XOR } y^2_1) \neq \prob(z_2\mid s_1, c_1, a_1) = \prob(x_1) = 0.5$.

\subsection{Proof of the generalized better reply secure property for the augmented stage-game}
\label{appendix:example}
We show that when $c > 24$ the augmented stage-game $\hat G_1$ in Section \ref{sec:example} is generalized better reply secure. For that matter, we set $\beta^*(q) = \mathds{1}(q \leq 1/3)$ and consider the following five cases.
\begin{enumerate}[{Case} (i), leftmargin=1.2cm]
    \item $r^0_1(\bar\alpha, \bar q) \neq 0$.
    In this case Bayes' rule doesn't hold at $(\bar\alpha, \bar q)$. We focus on agent $0$ and select the belief to satisfy Bayes' rule as follows:
    \begin{align}
        \phi^0(\tilde\alpha, \tilde q) = (
        \tilde \alpha_2 p + \tilde \alpha_1 (1-p), \tilde \alpha_2 (1 - p) + \tilde \alpha_1 p)
    \end{align}
    Then this $\phi^0$ is a closed correspondence. From this construction of $\phi^0$, we can pick $\epsilon > 0$ such that
    \begin{align*}
        r^0_1(\tilde\alpha, \phi^0(\tilde\alpha, \tilde q))
         = 0 > r^0_1(\bar\alpha, \bar q ) + \epsilon
    \end{align*}    
    \item $r^0_1(\bar\alpha, \bar q) = 0$, and $\bar\pi_{-1} \neq 1/3$ and $\bar\pi_1 \neq 1/3$.

    Since $\beta^*(q) = 1$ if $q < 1/3$, $\beta^*(q) = 0$ if $q > 1/3$, $\beta^*(\cdot)$ is continuous at points where $q \neq 1/3$. Hence, we can find $\epsilon > 0$ s.t. $\beta^*(\tilde q_{-1}) = \beta^*(\bar q_{-1})$ for all $\tilde q_{-1} \in (\bar q_{-1} - \epsilon, \bar q_{-1} + \epsilon)$, and $\beta^*(\tilde q_{1}) = \beta^*(\bar q_{1})$ for all $\tilde q_{1} \in (\bar q_{1} - \epsilon, \bar q_{1} + \epsilon)$. In this region we have 
    \begin{align}
        r^A_1(\alpha, \tilde q) = r^A_1(\alpha, \bar q) 
    \end{align}
    for all $\alpha$. Let 
    \begin{align}
        \phi^A(\tilde\alpha, \tilde q) = \argmax_\alpha r^A_1(\alpha, \tilde q)
    \end{align}
    Because $r^A_1(\cdot)$ is continuous in the region under consideration, $\phi^A(\cdot)$ has a closed graph from Berge's maximum theorem. Note that for all $\tilde q_{1} \in (\bar q_{1} - \epsilon, \bar q_{1} + \epsilon)$, $\tilde q_{1} \in (\bar q_{1} - \epsilon, \bar q_{1} + \epsilon)$
    \begin{align}
        r^A_1(\phi^A(\tilde\alpha, \tilde q), \tilde q) = \max_\alpha r^A_1(\alpha, \tilde q) = \max_\alpha r^A_1(\alpha, \bar q)
    \end{align}
    If $\max_\alpha r^A_1(\alpha, \bar q) > r^A_1(\bar\alpha, \bar q)$ we can find $\epsilon > 0$ such that for $\tilde q_{1} \in (\bar q_{1} - \epsilon, \bar q_{1} + \epsilon)$, $\tilde q_{1} \in (\bar q_{1} - \epsilon, \bar q_{1} + \epsilon)$, $ r^A_1(\phi^A(\tilde\alpha, \tilde q), \tilde q) = \max_\alpha r^A_1(\alpha, \tilde q) \geq r^A_1(\bar\alpha, \bar q) + \epsilon$.

    If $\max_\alpha r^A_1(\alpha, \bar q) = r^A_1(\bar\alpha, \bar q)$, then Alice has no profitable deviation. Furthermore, since $r^0_1(\bar\alpha, \bar q) = 0$, agent $0$ has no profitable deviation. Consequently, $(\bar\alpha, \bar q)$ is an equilibrium if $max_\alpha r^A_1(\alpha, \bar q) = r^A_1(\bar\alpha, \bar q)$.
    \item $r^0_1(\bar\alpha, \bar q) = 0$, $\bar\pi_{-1} = 1/3$ and $\bar\pi_1 \neq 1/3$.

    Note that $\bar q_{-1} = 0.8\bar\alpha_1 + 0.2\bar\alpha_2 = 1/3$ and $\beta^*(\bar q_{-1}) = 1/3$.
    Since $\bar\pi_1 \neq 1/3$, we can find $\epsilon > 0$ s.t. $\beta^*(\tilde q_{1}) = \beta^*(\bar q_{1})$ for all $\tilde q_{1} \in (\bar q_{1} - \epsilon, \bar q_{1} + \epsilon)$. Therefore, 
    \begin{align}
        r^A_1(\bar\alpha, \bar q) = 0.5c(1 - \bar\alpha_1 + \bar\alpha_2) + 0.5(2 - \bar\alpha_1 - \bar\alpha_2) + 0.5(3\bar q_{1} - 1) \beta^*(\bar q_1)
    \end{align}
    Pick for Alice
    \begin{align}
        \phi^A(\tilde\alpha, \tilde q) = (0, 1)
    \end{align}
    for all $\tilde \alpha_i \in (\bar \alpha_i - \epsilon, \bar \alpha_i + \epsilon), i=1,2$, $\tilde q_i \in (\bar q_i - \epsilon, \bar q_i + \epsilon), i=-1,1$. We get
    \begin{align}
        r^A_1(\phi^A(\tilde\alpha, \tilde q), \tilde q) = &c + 0.5 + 0.5( 0.6 - 1)\beta^*(\tilde q_{-1}) + 0.5(2.4 - 1)\beta^*(\tilde q_{-1})
        \notag\\
        = &c + 0.5 - 0.2\beta^*(\tilde q_{-1}) + 0.7\beta^*(\bar q_{-1})
    \end{align}
    and 
    \begin{align}
        & r^A_1(\phi^A(\tilde\alpha, \tilde q), \tilde q) - r^A_1(\bar\alpha, \bar q) - \epsilon  
        \notag\\
        = & 0.5c(1 + \bar\alpha_1 - \bar\alpha_2) -0.5(1 + \bar\alpha_1 + \bar\alpha_2) 
        \notag\\
        & -  0.2\beta^*(\tilde q_{-1}) + 0.5(2.4 - 3\bar q_1)\beta^*(\bar q_{-1}) - \epsilon
        \notag\\
        \geq & 0.5c(1 + \bar\alpha_1 - \bar\alpha_2) -0.5 * 3 - 0.2 - 0.5 * 0.6 - \epsilon
    \end{align}
    When $\bar q_{-1} = 1/3$, then $0.8\bar\alpha_1 + 0.2\bar\alpha_2 = 1/3 \Rightarrow \bar \alpha_1 = 5/12 - 3/12 \bar \alpha_2$. Therefore,
    \begin{align}
        1 + \bar \alpha_1 -\bar \alpha_2 = 17/12 - 15/12 \bar \alpha_2 \geq 1/6
    \end{align}
    where the minimum is at $\bar\alpha_1 = 1/6$ and $\bar\alpha_2=1$.

    When $c > 24$, then 
    \begin{align}
        0.5c(1 + \bar \alpha_1 -\bar \alpha_2) \geq c/12 > 2
    \end{align}
    and 
    $r^A_1(\phi^A(\tilde\alpha, \tilde q), \tilde q) - r^A_1(\bar\alpha, \bar q) - \epsilon > 0$.
    \item $r^0_1(\bar\alpha, \bar q) = 0$, and $\bar\pi_1 = 1/3$ and $\bar\pi_{-1} \neq 1/3$.

    This case is similar to case (iii). 
     Since $\bar\pi_{-1} \neq 1/3$, we can find $\epsilon > 0$ s.t. $\beta^*(\tilde q_{-1}) = \beta^*(\bar q_{-1})$ for all $\tilde q_{-1} \in (\bar q_{-1} - \epsilon, \bar q_{-1} + \epsilon)$. Furthermore, 
    \begin{align}
        & r^A_1(\bar\alpha, \bar q) 
        \notag\\
        = & 0.5c(1 - \bar\alpha_1 + \bar\alpha_2) + 0.5(2 - \bar\alpha_1 - \bar\alpha_2) + 0.5(3\bar q_{-1} - 1) \beta^*(\bar q_{-1})
    \end{align}
    Pick for Alice the closed correspondence (as in case (iii))
    \begin{align}
        \phi^A(\tilde\alpha, \tilde q) = (0, 1)
    \end{align}
    for all $\tilde \alpha_i \in (\bar \alpha_i - \epsilon, \bar \alpha_i + \epsilon), i=1,2$, $\tilde q_i \in (\bar q_i - \epsilon, \bar q_i + \epsilon), i=-1,1$. Then
    \begin{align}
        & r^A_1(\phi^A(\tilde\alpha, \tilde q), \tilde q) 
        \notag\\
        = &c + 0.5 - 0.2\beta^*(\bar q_{-1}) + 0.7\beta^*(\tilde q_{-1})
    \end{align}
     and 
    \begin{align}
        & r^A_1(\phi^A(\tilde\alpha, \tilde q), \tilde q) - r^A_1(\bar\alpha, \bar q) - \epsilon  
        \notag\\
        = & 0.5c(1 + \bar\alpha_1 - \bar\alpha_2) -0.5(1 + \bar\alpha_1 + \bar\alpha_2) 
        \notag\\
        & + 0.5( 0.6 - 3\bar q_{-1})\beta^*(\bar q_{-1}) + 0.7\beta^*(\tilde q_{-1}) - \epsilon
        \notag\\
        \geq & 0.5c(1 + \bar\alpha_1 - \bar\alpha_2) -0.5 * 3 - 0.5 * 2.4 - \epsilon
    \end{align}
    When $\bar q_{1} = 1/3$, $0.2\bar\alpha_1 + 0.8\bar\alpha_2 = 1/3 \Rightarrow \bar \alpha_2 = 5/12 - 3/12 \bar \alpha_1$. Therefore,
    \begin{align}
        1 + \bar \alpha_1 -\bar \alpha_2 = 7/12 + 15/12 \bar \alpha_1 \geq 7/12.
    \end{align}
     When $c > 24$, then 
    \begin{align}
        0.5c(1 + \bar \alpha_1 -\bar \alpha_2) \geq 7/24 c  > 2.7
    \end{align}
    and 
    $r^A_1(\phi^A(\tilde\alpha, \tilde q), \tilde q) - r^A_1(\bar\alpha, \bar q) - \epsilon > 0$.
    \item $r^0_1(\bar\alpha, \bar q) = 0$, and $\bar\pi_1 = 1/3$ and $\bar\pi_{-1} = 1/3$.

    We have
    \begin{align}
        r^A_1(\bar\alpha, \bar q) = 0.5c(1 - \bar\alpha_1 + \bar\alpha_2) + 0.5(2 - \bar\alpha_1 - \bar\alpha_2)
    \end{align}
    Pick for Alice the closed correspondence (as in cases (iii) and (iv))
    \begin{align}
        \phi^A(\tilde\alpha, \tilde q) = (0, 1)
    \end{align}
    for all $\tilde \alpha_i \in (\bar \alpha_i - \epsilon, \bar \alpha_i + \epsilon), i=1,2$, $\tilde q_i \in (\bar q_i - \epsilon, \bar q_i + \epsilon), i=-1,1$. Then
    \begin{align}
        & r^A_1(\phi^A(\tilde\alpha, \tilde q), \tilde q) - r^A_1(\bar\alpha, \bar q) - \epsilon  
        \notag\\
        = &0.5c(1 + \bar\alpha_1 - \bar\alpha_2) -0.5(1 + \bar\alpha_1 + \bar\alpha_2) -0.2\beta^*(\tilde q_{-1}) + 0.7\beta^*(\tilde q_{-1}) - \epsilon
        \notag\\
        \geq & 0.5c(1 + \bar\alpha_1 - \bar\alpha_2) -0.5 * 3 - 0.2 - \epsilon
    \end{align}
    Then we have $r^A_1(\phi^A(\tilde\alpha, \tilde q), \tilde q) - r^A_1(\bar\alpha, \bar q) - \epsilon  > 0$ following the steps in (iv).
\end{enumerate}